\documentclass[journal]{IEEEtran}
\usepackage{amsfonts}
\usepackage{amsmath}
\usepackage{graphics}
\usepackage{enumerate}
\usepackage{amssymb}
\usepackage{amsthm}
\newcommand{\appref}[1]{\hyperref[#1]{{Appendix~\ref*{#1}}}}
\newcommand{\be}{\begin{eqnarray} \begin{aligned}}
\newcommand{\ee}{\end{aligned} \end{eqnarray} }
\newcommand{\benn}{\begin{eqnarray*} \begin{aligned}}
\newcommand{\eenn}{\end{aligned} \end{eqnarray*}}

\newcommand*{\cE}{\mathcal{E}}
\newcommand*{\cF}{\mathcal{F}}
\newcommand*{\cG}{\mathcal{G}}
\newcommand*{\cH}{\mathcal{H}}

\newcommand*{\cL}{\mathcal{L}}

\newcommand*{\tr}{\mathop{\mathrm{tr}}\nolimits}

\newcommand{\bc}{\begin{center}}
\newcommand{\ec}{\end{center}}

\newtheorem{theorem}{Theorem}[section]
\newtheorem{lemma}[theorem]{Lemma}

\newtheorem{corollary}[theorem]{Corollary}

\usepackage{amsfonts}

\def\01{\{0,1\}}

\newcommand{\ket}[1]{|#1\rangle}
\newcommand{\bra}[1]{\langle#1|}

\newcommand{\proj}[1]{|#1\rangle\langle#1|}

\newcommand{\diag}{\mbox{diag}}

\newcommand*{\mylabel}[1]{\label{#1}}

\begin{document}
\title{The entropy power inequality\\
 for quantum systems}
\author{Robert K\"onig\thanks{R. K\"onig is with the Institute for Quantum Computing and the Department of Applied Mathematics, University of Waterloo, Waterloo, ON N2L 3G1, Canada, email {\sf rkoenig@uwaterloo.ca}.} and Graeme Smith\thanks{ G. Smith is with IBM TJ Watson Research Center, 1101 Kitchawan Road, Yorktown Heights, NY 10598, USA.}} 
\maketitle
\begin{abstract}When two independent analog signals, $X$ and $Y$
are added together giving $Z=X+Y$, the entropy of $Z$, $H(Z)$, is not
a simple function of the entropies $H(X)$ and
$H(Y)$, but rather depends on the details of $X$ and $Y$'s
distributions.  Nevertheless, the entropy power inequality
(EPI), which states that $e^{2H(Z)} \geq e^{2H(X)} + e^{2H(Y)}$, gives
a very tight restriction on the entropy of $Z$.  This inequality
has found many applications in information theory and statistics.  The
quantum analogue of
adding two random variables is the combination of two independent
bosonic modes at a beam splitter.  The purpose
of this work is to give a detailed outline of the proof of two
separate generalizations of the entropy power inequality to the
quantum regime.  Our proofs are similar in spirit to standard
classical proofs of the EPI, but some new quantities and ideas are
needed in the quantum setting.
Specifically, we find a new quantum de Bruijin identity relating
entropy production under  diffusion to a divergence-based
quantum Fisher information. Furthermore, this Fisher information exhibits  certain convexity properties in the context of beam splitters. 
\end{abstract}
\tableofcontents
\section{Introduction}
\subsection{The classical entropy power inequality: formulation}
In his 1948 paper~\cite{Shannon48}, Shannon identified the (differential) entropy 
\begin{align*}
H(X)=-\int f_X(x)\log f_X(x)dx
\end{align*} of a random variable~$X$ 
with support on~$\mathbb{R}^n$
 as the relevant measure for its information content. Using differential entropies, Shannon obtained an expression  for  the capacity of a noisy channel involving continuous signals. As a corollary, he  explicitly computed the 
capacity of the  white thermal noise  channel. These are   among his most well-known contributions to continous-variable classical information theory. However, Shannon also recognized another fundamental property of differential entropy  called the entropy power inequality (EPI).

Consider a situation where an input signal (or random variable)~$X$ is combined with an independent signal~$Y$, resulting in an output~$X+Y$ with probability density function given by the convolution 
\begin{align}
f_{X+Y}(z)&=\int  f_X(x)f_Y(z-x) d^nx\qquad\textrm{ for all }z\in\mathbb{R}^n\ .\mylabel{eq:convolutionclassical}
\end{align}
We may think of the  operation
\begin{align}
(f_X, f_Y)\mapsto f_{X+Y}\mylabel{eq:additionlawclassical}
\end{align}
as an `addition law' on the space of probability density functions.  Shannon's entropy power inequality relates the entropies of the input- and output variables in this process. It is commonly stated as
\begin{align}
e^{2H(X+Y)/n}\geq e^{2H(X)/n}+e^{2H(Y)/n}\ .\mylabel{eq:epiclassical}
\end{align}
It should be stressed that inequality~\eqref{eq:epiclassical} holds without assumptions on the particular form of the probability density functions~$f_X$, $f_Y$ beyond certain regularity conditions. 
Specialized to the case  where one of the arguments in~\eqref{eq:additionlawclassical}, say, $f_Y$, is a fixed probability density function (e.g., a Gaussian), inequality~\eqref{eq:epiclassical} immediately gives a lower bound on the output entropy of the so-called additive noise channel $X\mapsto X+Y$ (e.g., the white thermal noise channel) in terms of the entropy~$H(X)$ of the input signal.  Given this fact and the fundamental nature of the operation~\eqref{eq:additionlawclassical},  it is is hardly suprising that the EPI~\eqref{eq:epiclassical} has figured prominently in an array of information-theoretic works since its introduction in~1948. For example, it has been used to obtain converses for the capacity of the Gaussian broadcast channel, and to show convergence in relative entropy for the central limit theorem~\cite{Barron86central}. Other uses of the entropy power inequality in multi-terminal information theory may be found e.g., in~\cite{ElGamal02}.

The expression entropy power stems from the fact that a  centered normal (or Gaussian) distribution~$f_X(x)=\frac{1}{(2\pi\sigma^2)^{n/2}}e^{-\|x\|^2/(2\sigma^2)}$ on $\mathbb{R}^n$ with variance~$\sigma^2$ has entropy~$H(X)=\frac{n}{2}\log 2\pi e\sigma^2$. Hence
the quantity $\frac{1}{2\pi e}e^{2H(X)/n}$ is the variance of $X$, commonly referred to as its power or average energy.  For a general random variable~$Z$, the quantity~$\frac{1}{2\pi e}e^{2H(Z)/n}$  is the power of a Gaussian variable with the same  entropy as~$Z$. This fact has been connected~\cite{CostaThomas84}  to Brunn-Minkowski-type inequalities bounding the volume of the set-sum of two sets in terms of their spherical counterparts, which in turn inspired generalizations to free probability theory~\cite{SzarekVoicu96}. In the following, we will omit the factor~$\frac{1}{2\pi e}$ and refer to~$e^{2H(Z)/n}$ as the {\em entropy power of~$Z$}.

The entropy power inequality~\eqref{eq:epiclassical} can be reeexpressed as a kind of `concativity' property of the entropy power under a slightly modified `addition rule'~$(X,Y)\mapsto X\boxplus_{1/2}Y$.  To define the latter, let 
$\mu X$  denote the random variable with probability density
$f_{\mu X}(x)=\frac{1}{\mu}f_X(\frac{x}{\mu})$, $x\in \mathbb{R}^n$
for any scalar $\mu>0$. We define $X\boxplus_{1/2}Y$ as the result of rescaling both~$X$ and~$Y$ by~$1/\sqrt{2}$  and subsequently taking their 
sum (see Eq.~\eqref{eq:additionlawclassicalcov} below). Because differential entropies satisfy the scaling property~$H(\mu X)=H(X)+n\log \mu$ for  $\mu>0$, this leads to the following reformulation of the entropy power inequality~\eqref{eq:epiclassical}:
\begin{align}
\tag{cEPI} 
e^{2H(X\boxplus_{1/2} Y)/n}\geq \frac{1}{2}e^{2H(X)/n}+
\frac{1}{2}e^{2H(Y)/n}\ .\mylabel{eq:epiclassicalcov}
\end{align}
Observe that the rhs.~is the equal-weight average of the entropy powers of $X$ and $Y$.

It turns out that the entropy $H(X)$ itself satisfies a concativity inequality analogous to~\eqref{eq:epiclassicalcov} (see~\eqref{eq:liebversion} below). More generally, introduce a parameter $0<\lambda<1$ controlling the  relative weight of~$X$ and~$Y$ and consider the  addition rule
\begin{align}
(f_X,f_Y)&\mapsto f_{X\boxplus_\lambda Y}\ \textrm{ where }\   f_{X\boxplus_\lambda Y}=f_{\sqrt{\lambda}X+\sqrt{1-\lambda}Y}\ .  \mylabel{eq:additionlawclassicalcov}
\end{align}
Note that this map is covariance-preserving when applied to two random-variables~$X$ and~$Y$ whose first moments and covariances  coincide, a property not shared by convolution without rescaling (i.e., ~\eqref{eq:additionlawclassical}). 

 The concativity inequality for the entropy takes the form
\begin{align}\tag{cEPI$^\prime$}
H(X\boxplus_\lambda Y)\geq \lambda H(X)+(1-\lambda)H(Y)\qquad\textrm{ for } 0<\lambda<1\ .\mylabel{eq:liebversion}
\end{align}
The set of inequalities~\eqref{eq:liebversion} can be shown~\cite{Lieb78,VerduGuo06} to be equivalent to~\eqref{eq:epiclassicalcov}. We will refer to both~\eqref{eq:epiclassicalcov},~\eqref{eq:liebversion} as entropy power inequalities; both express a type of concativity under the addition law~\eqref{eq:additionlawclassicalcov}.

\subsection{The quantum entropy power inequality: formulation}
Here we formulate and subsequently outline a proof of a quantum version of the entropy power inequality. It is phrased in terms of  the von Neumann entropy
\begin{align*}
S(X)&=-\tr(\rho_X\log\rho_X)
\end{align*}
of a state~$\rho_X$ on a separable infinite-dimensional Hilbert space describing $n$~bosonic degrees of freedom. In quantum optics terminology, it gives a relation between the entropies of states going in and coming out of a beamsplitter.   Central to this program is the identification of the correct definitions applying to the quantum case. This is what we consider our main contribution.

 Our `quantum addition law' generalizes~\eqref{eq:additionlawclassicalcov} and  is most conveniently described for the case of  two single-mode $(n=1)$ states~$\rho_X$ and $\rho_Y$ (although the generalization to $n>1$~modes is straightforward, see Section~\ref{sec:beamsplitterdef}). It takes the following form: the Gaussian unitary~$U_\lambda$ corresponding to a beam-splitter with transmissivity~$0<\lambda<1$ is applied to the two input modes. This is followed by tracing out the second mode. Formally, we have (see Section~\ref{sec:beamsplitterdef} for more precise definitions) 
\begin{align}
\begin{matrix}
(\rho_X,\rho_Y)\mapsto \rho_{X\boxplus_\lambda Y}=\cE_\lambda(\rho_X\otimes\rho_Y)\qquad\\
\qquad\textrm{ where } \cE_\lambda(\rho)=\tr_2 U_\lambda \rho U_\lambda^\dagger\ .
\end{matrix}~\mylabel{eq:additionlawquantumcov}
\end{align}
This choice of `addition law' is motivated by the fact that the resulting position- and momentum operators~$(Q,P)$ of the output mode in the Heisenberg picture are given by
\begin{align}
\begin{matrix}
Q&=\sqrt{\lambda }Q_1+\sqrt{1-\lambda}Q_2\\
P&=\sqrt{\lambda }P_1+\sqrt{1-\lambda}P_2\ 
\end{matrix}\mylabel{eq:XYlambdaaddquantum}
\end{align}
in terms of the original operators of the first $(Q_1,P_1)$ and second $(Q_2,P_2)$ mode. Hence~\eqref{eq:XYlambdaaddquantum} indeed mimics the classical addition rule~\eqref{eq:additionlawclassicalcov} in phase space.

Having defined  an addition law, we can state the generalizations of~\eqref{eq:epiclassicalcov} and~\eqref{eq:liebversion}
  to the quantum setting. They   are
\begin{align}
e^{S(X\boxplus_{1/2} Y)/n}&\geq \frac{1}{2}e^{S(X)/n}+\frac{1}{2}e^{S(Y)/n}\ ,\mylabel{eq:epiquantumcov}\tag{qEPI}\\
S(X\boxplus_\lambda Y)&\geq \lambda S(X)+(1-\lambda)S(Y)\qquad\textrm{ for } \lambda\in [0,1]\ .\tag{qEPI$^{\prime}$}\mylabel{eq:liebversionquantum}
\end{align}
Observe that~\eqref{eq:epiquantumcov} differs from~\eqref{eq:epiclassicalcov} by the absence of a factor~$2$ in the exponents. This  may be attributed to the fact that the phase space of $n$~bosonic modes is $2n$-dimensional. We emphasize, however, that neither~\eqref{eq:epiquantumcov} nor~\eqref{eq:liebversionquantum} appear to follow in a straightforward manner from their classical counterparts. A further distinction from the classical case is that, to the best of our knowledge, Eq.~\eqref{eq:epiquantumcov} and~\eqref{eq:liebversionquantum} do not appear to be equivalent. Note also that Guha established~\eqref{eq:liebversionquantum} for the special case~$\lambda=1/2$ using different techniques~\cite[Section 5.4.2]{GuhaPhD}.

 We will not discuss implications of these inequalities here (but see~\cite{KoeGrae12} for an upper bound on classical capacities), but expect them to be widely applicable.
We point out that the importance of a quantum analog of the entropy power inequality in the context of additive noise channels was recognized in ealier work by Guha et al.~\cite{Guhaetal07}. These authors proposed a different generalization of~\eqref{eq:epiclassicalcov}, motivated by the fact that the average photon number~$\tr(a^\dagger a\rho)$ of a state~$\rho$ is the natural analog of the variance of a (centered) probability distribution. In analogy to the definition of (classical) entropy power, they defined the  photon number of a state~$\rho$ as the average photon number of a Gaussian state~$\sigma$ with identical  entropy. Their conjectured photon number inequality states that this quantity is concave under the addition rule~$(X,Y)\mapsto X\boxplus_\lambda Y$. The photon number inequality has been shown to hold in the case of Gaussian states~\cite{GuhaPhD}. 

We find that our more literal generalization~\eqref{eq:epiquantumcov}  (for $\lambda=1/2$) allows us to translate a proof of the classical entropy power inequality to the quantum setting. The formulation~\eqref{eq:epiquantumcov} also yields tight bounds on classical capacities~\cite{KoeGrae12} for thermal noise associated with the transmissivity-$1/2$-beamsplitter. We conjecture that the quantity $e^{S(X)/n}$ is concave with respect to the addition rule~$\boxplus_\lambda$ for all $0<\lambda<1$ (a fact which is trivial in the classical case due to the scaling property of entropy). If shown to be true, this establishes tight bounds on the classical capacity of thermal noise channels, for all transmissivities.

\subsection{The classical entropy power inequality: proof sketch\mylabel{sec:classicalepiproofreview}}
It is instructive
to review the basic elements of the proof of the classical entropy power inequality. 
Roughly, there are two different proof methods: the first one relies on Young's inequality in the version given by Beckner~\cite{Beckner75}, and Brascamp and Lieb~\cite{BrascampLieb76}. The second one relies on the relation between entropy and Fisher information, known as de Bruijin's identity. For a detailed account of the general form of these two approaches, we refer to~\cite{Demboetal91}. While~\cite{Demboetal91} gives a good overview, it is worth mentioning that more recently,  Rioul has found proofs of the EPI that manage to rely entirely on the use of `standard' information measures such as mutual information~\cite{liuvis05,rioultwo}. His work also involves ideas from the second approach.

Here we focus on the second type of proof, originally introduced by Stam~\cite{Stam59} and further simplified and made more rigorous by Blachman~\cite{Blachman65} and others~\cite{Barron86central,Demboetal91}.   Its main components can be summarized as follows:
\begin{enumerate}[(i)]
\item{\bf Gaussian perturbations:} For a random variable~$X$ on $\mathbb{R}^n$ define 
\begin{align}
X_t=X+\sqrt{t}Z\qquad t\geq 0\mylabel{eq:xtevolveddefinition}
\end{align}
where $Z$ be a standard normal with variance $\sigma^2=1$.
 It is well-known that the family of distributions $f_t=f_{X_t}$ satisfies the heat or diffusion equation
\begin{align}
\frac{\partial}{\partial t} f_t &=\Delta f_t\ ,\qquad \Delta=\sum_{j=1}^N \frac{\partial^2}{\partial x^2_j}\ \mylabel{eq:heatequation}
\end{align}
with initial condition~$f_0(x)=f_X(x)$~for all $x\in\mathbb{R}^N$.   In particular, diffusion acts as a one-parameter semigroup on the set of probability density functions, i.e., 
\begin{align}
X_{t_1+t_2}=(X_{t_1})_{t_2}\qquad\textrm{ for all }t_1,t_2\geq 0\ .\label{eq:semigrouppropertyt1t2}
\end{align}

Time evolution under~\eqref{eq:heatequation} smoothes out initial spatial inhomogenities. In the limit~$t\rightarrow \infty$, the distribution of~$X_t$ becomes `simpler': it approaches a Gaussian, and the asymptotic scaling of its entropy~$H(X_t)\sim g(t)$ is  independent of the initial distribution~$f_X$. This is crucial in the  proof of the EPI.

A second essential feature of the evolution~\eqref{eq:heatequation} is that it commutes with convolution: we have (for any $0<\lambda<1$ and $t>0$)
\begin{align}
X_t\boxplus_\lambda Y_t\equiv (X\boxplus_\lambda Y)_t \mylabel{eq:diffusingoutputequivdiffusinginput}
\end{align}
for the time-evolved versions $(X_t,Y_t)$ of two independent random variables $(X,Y)$.

\item{\bf de Bruijin's identity: } A third property of the diffusion process~\eqref{eq:xtevolveddefinition} concerns the rate of entropy increase for infinitesimally small times, $t\rightarrow 0$. de Bruijin's identity relates this rate to a well-known estimation-theoretic quantity. It states that
\begin{align}
\begin{matrix}
\frac{d}{dt}\Big|_{t=0} H(X+\sqrt{t}Z) =\frac{1}{2}J(X)\qquad \textrm{ where }\\
\qquad \qquad J(X)=\int (\nabla f(x))^T(\nabla f(x))\cdot \frac{1}{f(x)}d^Nx\ .\ 
\end{matrix}\mylabel{eq:debruijiclassicalstated}
\end{align}
The quantity~$J(X)$ is the {\em Fisher information} $J(X)=J(f^{(\theta)};\theta)|_{\theta=0}$ associated with the family~$\{f^{(\theta)}\}_{\theta\in\mathbb{R}}$ of distributions
\begin{align}
f^{(\theta)}(x)&=f(x-\theta\cdot \vec e)\qquad \textrm{ for all }x\in\mathbb{R}^n\mylabel{eq:translationsbytheta}\ .
\end{align}
where $\vec{e}=(1,\ldots,1)$. 
In other words, the random variables $\{X^{(\theta)}\}_{\theta}$ are obtained by translating~$X$ by an unknown amount~$\theta$ in each direction, and~$J(X)$ relates to the problem of estimating~$\theta$ (by the Cram\'er-Rao bound). More importantly, the quantity~$J(X)$ inherits various general properties of Fisher information, generally defined as
\begin{align} 
J(f^{(\theta)};\theta)\big|_{\theta=\theta_0}&=\int \left(\frac{\partial}{\partial \theta}\log f^{(\theta)}(x)\right)^2 f^{(\theta)}(x) dx\Big|_{\theta=\theta_0} \mylabel{eq:fisherinformationgeneral}
\end{align}
for a parameterized family $\{f^{(\theta)}\}$ of probability distributions on~$\mathbb{R}$ (this is replaced by a matrix in the case $n>1$).

The link between entropy and Fisher information 
expressed by de Bruijin's identity is appealing in itself. Its usefulness in the proof of the EPI stems from yet another compatibility property with convolution. Translations as defined by~\eqref{eq:translationsbytheta} can be equivalently applied before or after applying addition (convolution), that is, we have 
\begin{align}
(X\boxplus_{\lambda} Y)^{(\theta)}&=X^{(\sqrt{\lambda}\theta)}\boxplus_{\lambda} Y^{(\sqrt{1-\lambda}\theta)}\label{eq:additiontranslationconv}
\end{align}
for $0<\lambda<1$ and $\theta\in\mathbb{R}$. 
In other words, adding (convolving) $X$ and $Y$ and then translating by~$\theta$ is equivalent to adding up translated versions $X^{(\sqrt{\lambda}\theta)},Y^{(\sqrt{1-\lambda}\theta)}$. 

\item{\bf Convexity inequality for Fisher information}: This shows that Fisher information behaves in a `convex' manner under the addition law~\eqref{eq:additionlawclassicalcov}, i.e., 
\begin{align}
 J(X\boxplus_{\lambda}Y)\leq \lambda J(X)+(1-\lambda) J(Y) \ .\mylabel{eq:fisherinformationinequalitycomplete}
\end{align}
A related inequality for the addition law when $\lambda=1/2$ is known as Stam's inequality; it takes the form
\begin{align}
\frac{2}{J(X\boxplus_{1/2}Y)}\geq \frac{1}{J(X)}+\frac{1}{J(Y)}\ .\mylabel{eq:staminequalityclassical}
\end{align}
Proofs of such inequalities have been given in~\cite{Stam59,Blachman65,CostaThomas84}. An insightful proof was provided by Zamir~\cite{Zamir98}, relying only on  basic properties of Fisher information. The most fundamental of those is the information-processing inequality, which states that  Fisher information is non-increasing under the application of a channel (stochastic map)~$\cE$, 
\begin{align} 
J(\cE(f^{(\theta)});\theta)\big|_{\theta=\theta_0}\leq J(f^{(\theta)};\theta)\big|_{\theta=\theta_0}\ .\mylabel{eq:dataprocessingfisherinfo}
\end{align}
The Fisher information inequalities~\eqref{eq:fisherinformationinequalitycomplete} and~\eqref{eq:staminequalityclassical} are  simple consequences of this data-processing inequality and the compatibility~\eqref{eq:additiontranslationconv} of  convolution with translations. For example, the quantities on the two sides of~\eqref{eq:fisherinformationinequalitycomplete} are the Fisher informations of the two families of distributions~$\{X^{(\theta)}\}_{\theta}$ and~$\{(X^{(\sqrt{\lambda}\theta)},Y^{(\sqrt{1-\lambda}\theta)})\}_\theta$, and are directly associated with both sides of the compatiblity identity~\eqref{eq:additiontranslationconv}.
\end{enumerate}
The entropy inequalities~\eqref{eq:epiclassicalcov} and~\eqref{eq:liebversion}   follow immediately from~(i)--(iii). Concretely, let us describe the argument of Blachman~\cite{Blachman65} adapted to the proof of~\eqref{eq:liebversion}
(The derivation of~\eqref{eq:epiclassicalcov} is identical in spirit, but slightly more involved. It uses~\eqref{eq:staminequalityclassical} instead of~\eqref{eq:fisherinformationinequalitycomplete}.). Note that a similar proof which avoids the use of asymptotically divergent quantities is provided in~\cite{Demboetal91}, yet we find that the following version is more amenable to a quantum generalization.

The core idea is to consider  perturbed or diffused versions
\begin{align*}
X_t &=X+\sqrt{t}Z_1\\
Y_t&=Y+\sqrt{t}Z_2\ ,
\end{align*}
of the orignal random variables $(X,Y)$, where $t\geq 0$ and where $Z_1,Z_2$ are independent  random variables with standard normal distribution. This  is useful because asymptotically as~$t\rightarrow\infty$, the entropy power inequality is satisfied trivially for the pair~$(X_t,Y_t)$ by~$(i)$. We call this the infinite-diffusion limit. de Bruijin's identity~$(ii)$ and the Fisher information inequality~$(iii)$ provide an understanding of the involved entropies in the opposite limit, where the diffusion only acts for infinitesimally short times. This allows one to extend the result from infinitely large times down to all values of~$t$ including~$0$. In the latter case, there is no diffusion and we are dealing with the original random variables $(X,Y)$.

\begin{proof}[Sketch proof of~\eqref{eq:liebversion}]
In more detail, introduce the quantity
\begin{align}
\delta(t):= H(X_t\boxplus_\lambda Y_t)-\lambda H(X_t)-(1-\lambda)H(Y_t)
\end{align}
measuring the  amount of potential violation  of the entropy power inequality for the pair~$(X_t,Y_t)$. The entropy power inequality~\eqref{eq:liebversion} for~$(X_t,Y_t)$ is equivalent to the statement~$\delta(t)\geq 0$. Since $(X_0,Y_0)\equiv (X,Y)$, the entropy power inequality~\eqref{eq:liebversion} amounts to the inequality
\begin{align}
\delta(0)\geq 0\ .\mylabel{eq:epiatzerodef}
\end{align}

To show~\eqref{eq:epiatzerodef}, observe that in the infinite-diffusion limit~$t\rightarrow\infty$, we have
\begin{align}
\lim_{t\rightarrow\infty}\delta(t)&=0\mylabel{eq:infinitediffusezero}
\end{align}
as all entropies $H(X_t)\sim H(Y_t)\sim H(X_t\boxplus_\lambda Y_t)$ have the same asymptotic scaling. Here we used the compatibility~\eqref{eq:diffusingoutputequivdiffusinginput} of 
diffusion and the addition rule~\eqref{eq:additionlawclassicalcov}.
 de Bruijin's identity~$(i)$ and the convexity inequality for Fisher information~$(ii)$ imply that 
the function $\delta(t)$ is non-increasing, i.e, 
\begin{align}
\dot{\delta}(t)\leq 0\qquad\textrm{ for all } t\geq 0\ . \mylabel{eq:deltaderivnonneg}
\end{align}
Indeed, we can use the  semigroup property~\eqref{eq:semigrouppropertyt1t2}
 of diffusion to show that it is sufficient to consider infinitesimal perturbations, i.e., the derivative of $\delta(t)$ at~$t=0$. But the condition
\begin{align*}
\dot{\delta}(0)\leq 0 
\end{align*}
is simply the Fisher information inequality~\eqref{eq:fisherinformationinequalitycomplete}, thanks to de Bruijin's identity~\eqref{eq:debruijiclassicalstated} and the compatibility~\eqref{eq:diffusingoutputequivdiffusinginput} of diffusion with convolution.
Hence one obtains~\eqref{eq:deltaderivnonneg} which together with Eq.~\eqref{eq:infinitediffusezero} implies the claim~\eqref{eq:epiatzerodef}. 
\end{proof}

\subsection{The quantum entropy power inequality: proof elements\mylabel{sec:qpowerinequalitysketch}}
Here we argue that a careful choice of
corresponding quantum-mechanical concepts guarantees that the remarkable compatibility properties between convolution, diffusion, translation, entropy and Fisher information all have analogous quantum formulations. This yields a proof of the quantum entropy power inequalities~\eqref{eq:epiquantumcov} and~\eqref{eq:liebversionquantum} that follows exactly the structure of the classical proof, yet involves fundamentally quantum-mechanical objects. 

In this section, we briefly motivate and discuss the choices required. 

\begin{enumerate}[(i)]
\item{\bf A quantum diffusion process:} For a state $\rho$ of $n$~modes, we define 
\begin{align}
\rho_t  &=e^{t\cL}(\rho)\qquad t\geq 0\ \label{eq:rhotdefinitionbasic}
\end{align}
as the result of evolving for time~$t$ under the one-parameter semigroup $\{e^{t\cL}\}_{t\geq 0}$ generated by the Markovian Master equation
\begin{align}
\begin{matrix}
\frac{d}{d t}\rho_t =\cL(\rho_t) \\
\cL(\cdot) =- \frac{1}{4}\sum_{j=1}^n \left([Q_j,[Q_j,\cdot]]+[P_j,[P_j,\cdot]]\right)\ .
\end{matrix}\label{eq:markovmaster}
\end{align}
This choice of Liouvillean~$\cL$ is motivated (as in~\cite{Hall00}) by the fact that it is the natural quantum counterpart of the (classical) Laplacian (cf.~\eqref{eq:heatequation})
\begin{align*}
\Delta(\cdot)=\frac{1}{4}\sum_{j=1}^n \left(\{q_j,\{q_j,\cdot\}\}+\{p_j,\{p_j,\cdot\}\}\right)
\end{align*} acting on probability density function on the  symplectic phase space~$\mathbb{R}^{2n}$ of  $n$~classical particles (the factor~$\frac{1}{4}$ is introduced here for mathematical convenience). Indeed, $\cL$ is obtained from the Dirac correspondence $\{\cdot,\cdot\}\mapsto \frac{1}{i\hbar}[\cdot,\cdot]$ between the Poisson bracket and the commutator, as well as phase space variables and canonical operators, $(q_j,p_j)\mapsto (Q_j,P_j)$. 

 Alternatively, one can argue that the Wigner functions~$\{W_t\}$ of the states~$\{\rho_t\}$ defined by the evolution~\eqref{eq:rhotdefinitionbasic} obey the heat equation~\eqref{eq:heatequation} on~$\mathbb{R}^{2n}$. 

Hall~\cite{Hall00} has previously examined the relation between the Master equation~\eqref{eq:markovmaster} and classical Fisher information. He argued that the position- and momentum density functions~$f_t(x)=\bra{x}\rho_t\ket{x}$ and $f_t(p)=\bra{p}\rho_t\ket{p}$ obey the heat equation~\eqref{eq:heatequation}. Correspondingly, he obtained de Bruijin-type identities for the measurement-entropies when measuring~$\rho_t$ in the (overcomplete) position- and momentum eigenbases, respectively. These results are, however, insufficient for our purposes as they refer to purely classical objects (derived from quantum states).

\item{\bf Divergence-based quantum Fisher information:}
The information processing inequality~\eqref{eq:dataprocessingfisherinfo} is well-known to be  the defining property of classical Fisher information: A theorem by \v{C}encov~\cite{Cencov81,Campbell86} singles the Fisher information metric as the unique monotone metric (under stochastic maps) on the manifold of (parametrized) probability density functions. Unfortunately, such a characterization does not hold in the quantum setting: Petz~\cite{Petz96} established a one-to-one correspondence  between monotone  metrics for quantum states and the set of matrix-monotone functions. 

It turns out that the correct choice of quantum Fisher information in this context is given by the second derivative of the relative entropy or divergence $S(\rho\|\sigma)=\tr(\rho\log\rho-\rho\log\sigma)$, that is, 
\begin{align}
J(\rho^{(\theta)};\theta)\Big|_{\theta=\theta_0}&=\frac{d^2}{d\theta^2} S(\rho^{(\theta_0)}\|\rho^{(\theta)})\Big|_{\theta=\theta_0}\ .\label{eq:qfishinfo}
\end{align}
This choice is guided by the fact that 
the classical Fisher information~\eqref{eq:fisherinformationgeneral} can also be written in this form, that is,
$J(f^{(\theta)};\theta)\big|_{\theta=\theta_0}=\frac{d^2}{d\theta^2} D(f^{(\theta_0)}\|f^{(\theta)})\Big|_{\theta=\theta_0}$,
where $D(f\|g)=\int f(x)\log\frac{f(x)}{g(x)}$ is the classical relative entropy.  We call the quantity~\eqref{eq:qfishinfo} the {\em divergence-based (quantum) Fisher information}. 

Lesniewski and Ruskai~\cite{LesRus99} have studied metrics obtained by differentiation of a `contrast functional'~$S(\cdot\|\cdot)$ in a similar manner as~\eqref{eq:qfishinfo}. They found that every monotone Fisher information arises from a quasi-entropy in this way. In particular, the data-processing inequality~\eqref{eq:dataprocessingfisherinfo} is a consequence of their general result. We will discuss  an independent elementary proof of the data processing inequality for the divergence-based Fisher information in Section~\ref{sec:divergencebasedqfisher}.

\item{\bf Translations of quantum states:}  These are simply translations in phase space. We use translated versions of a state~$\rho$  displaced in each direction, 
\begin{align*}
\rho^{(\theta,Q_j)}&=e^{i\theta P_j}\rho e^{-i\theta P_j}\\
\rho^{(\theta,P_j)}&=e^{-i\theta Q_j}\rho e^{i\theta Q_j}\qquad\textrm{ for }j=1,\ldots,n\ . 
\end{align*}
Clearly, this choice is dictated by the Heisenberg action on the mode operators,
\begin{align*}
e^{i\theta P_j}Q_je^{-i\theta P_j}&=Q_j+\theta\\
e^{-i\theta Q_j}P_je^{i\theta Q_j}&=P_j+\theta\ , 
\end{align*}
in analogy to the addition law~\eqref{eq:XYlambdaaddquantum}.

\end{enumerate}
In the remainder of this paper, we verify that  these definitions
indeed satisfy the necessary properties to provide a proof of the
quantum inequalities~\eqref{eq:epiquantumcov}
and~\eqref{eq:liebversionquantum}.
A word of mathematical caution is in order here:  several times we
will need to interchange a derivative with an integral or infinite
sum,
typically taking a derivative inside a trace.  This can be justified
as long as the functions involved are sufficiently smooth.
In order to focus on the main ideas in our argument, rather
than pursue a detailed justification of the interchange of limits we
will simply restrict our attention to families of states
that are sufficiently smooth.  This involves little loss of generality
since we are interested in the entropy of states satisfying an energy
bound.
On this set of states, entropy is continuous, so one can hope to
approximate non-smooth functions with smooth ones to obtain the
unrestricted result\footnote{Alternatively, one could introduce a
high-photon-number cutoff (well higher than the energy bound
of the states under consideration) to make the state space finite,
where these problems disappear, and take the limit as the cutoff
grows.}.
As a result, smoothness requirements are not usually considered in
proofs of the classical EPI~\cite{Stam59,Blachman65,liuvis05,VerduGuo06}
or studies of related quantum ideas~\cite{Giovannettietal04,Giovannetti10}. Developing a more mathematically thorough version of these arguments  appears to be a formidable task (even in the classical case, see~\cite{Barron86central}) which is left to future work.

\section{Basic definitions}
\subsection{Continuous-variable quantum information}
A quantum state of $n$~modes has $2n$~canonical degrees of freedom. Let  $Q_k$ and $P_k$ be the ``position'' and ``momentum'' operators of the $k$-th mode, for $k=1,\ldots,n$, acting on the Hilbert space~$\cH_n$ associated with the system.  Defining $R=(Q_1,P_1,\ldots,Q_n,P_n)$, these operators obey the commutation relations
\begin{align*}
[R_k,R_\ell]=iJ_{k\ell}\qquad\textrm{ where }\qquad J=\left(\begin{matrix}
0 & 1\\
-1 & 0
\end{matrix}\right)^{\oplus n}\ ,
\end{align*}
with $^{\oplus n}$ denoting the $n$-fold direct sum. 
The {\em Weyl displacement operators} are defined by
\begin{align*}
D(\xi)&=e^{i\xi\cdot J R}\qquad\textrm{ for }\xi\in\mathbb{R}^{2n}\ .
\end{align*}
These are unitary operators satisfying the relations
\begin{align*}
D(\xi)D(\eta)&=e^{-\frac{i}{2}\xi\cdot J\eta}D(\xi+\eta)\ 
\end{align*}
and 
\begin{align}
D(\xi)R_kD(\xi)^\dagger&=R_k+\xi_k I\ \qquad\textrm{ or } \mylabel{eq:displacementpropone}\\
[D(\xi),R_k]&=\xi_k D(\xi)\ .\mylabel{eq:displacementproperty}
\end{align}
This explains the terminology. Any state~$\rho$ takes the form
\begin{align}
\rho&=(2\pi)^{-n}\int \chi_\rho(\xi)D(-\xi)d^{2n}\xi\ ,\mylabel{eq:rhocharacteristicfunction}
\end{align}
where
\begin{align}
\chi_\rho(\xi)=\tr(D(\xi)\rho)\mylabel{eq:characteristicfunctiondef}
\end{align} 
is the {\em  characteristic function} of~$\rho$. The characteristic function~$\chi_\rho$ is well-defined for any trace-class operator, but the map~$\rho\mapsto \chi_\rho$ can be extended to all Hilbert-Schmidt class operators~$\rho$ by continuity. The map  is an isometry between the Hilbert space of these operators and the space~$L^2(\mathbb{R}^{2n})$ of square-integrable functions on the phase space:
\begin{align}
\tr(\rho^\dagger \sigma)&=(2\pi)^{-n}\int \overline{\chi_\rho(\xi)}\chi_\sigma(\xi)\qquad\textrm{ for all }\rho,\sigma\ .\label{eq:overlapcomputationcharacteristic}
\end{align}
Observe that for Hermitian $\rho=\rho^\dagger$, we have~$\chi_\rho(\xi)=\overline{\chi_\rho(-\xi)}$ since~$D(\xi)^\dagger=D(-\xi)$

A {\em Gaussian state} $\rho$ is uniquely characterized by the vector $d=(d_1,\ldots,d_{2n})\in\mathbb{R}^{2n}$ of first moments 
\begin{align}
d_k&=\tr(\rho R_k)\mylabel{eq:firstmomentvector}
\end{align} and its covariance matrix~$\gamma$ with entries
\begin{align}
\gamma_{k\ell}&=\tr\left(\rho\{R_k-d_k,R_\ell-d_\ell\}\right)\ ,\mylabel{eq:cmdef}
\end{align}
where $\{A,B\}=AB+BA$. 
Its characteristic function is 
\begin{align}
\chi(\xi)&=\exp\left(i\xi\cdot J d-\xi\cdot J^T\gamma J\xi/4\right)\ .\mylabel{eq:characteristicfunctiongaussian}
\end{align}
The covariance matrix~$\gamma$ of a quantum state satisfies the operator inequality
\begin{align}
\gamma \geq iJ\ ,\mylabel{eq:heisenberguncertaintyprinciple}
\end{align}
which is Heisenberg's uncertainty relation. Conversely, any matrix~$\gamma$ satisfying~\eqref{eq:heisenberguncertaintyprinciple} and vector $d\in\mathbb{R}^{2n}$ define a Gaussian quantum state via~\eqref{eq:characteristicfunctiongaussian}. Complex conjugating~\eqref{eq:heisenberguncertaintyprinciple} and adding it to itself shows that the covariance matrix~$\gamma\geq 0$ is positive definite. According to Williamson's theorem~\cite{Williamson36}, 
there is a symplectic matrix~$S\in Sp(2n,\mathbb{R})$, i.e., a matrix satisfying~$SJS^T=J$, such that
\begin{align}
S\gamma S^T&=\diag(\nu_1,\nu_1,\nu_2,\nu_2,\ldots,\nu_n,\nu_n)\ \textrm{ with }\nu_j\geq 0 \mylabel{eq:williamsonform}
\end{align}
is diagonal. Eq.~\eqref{eq:williamsonform} is called the {\em Williamson normal form} of~$\gamma$, and~$\{\nu_j\}_j$ are referred to as the {\em Williamson eigenvalues}.

Symplectic transformations~$S$ are in one-to-one correspondence with linear transformations $R\mapsto R':=SR$ preserving the canonical commutation relations. Each $S\in Sp(2n,\mathbb{R})$  defines a unitary~$U_S$ on~$\cH_n$ which realizes this map by conjugation, i.e., 
\begin{align}
U_S^\dagger R_kU_S &=\sum_{\ell=1}^{2n}S_{k\ell} R_{\ell} =:R'_k\ .\mylabel{eq:unitarymapweyl}
\end{align}
The unitary $U_S$ is unique up to a global phase.

If $S\in Sp(2n,\mathbb{R})$ brings the covariance matrix of a Gaussian state~$\rho$ into  Williamson normal form as in~\eqref{eq:williamsonform}, then~$\{R'_k\}_k$ are the {\em eigenmodes} of~$\rho$.
In terms of the creation, annihilation and number operators
\begin{align*}
a^\dagger_k=\frac{1}{\sqrt{2}}(Q_k'+iP_k')\qquad a_k=\frac{1}{\sqrt{2}}(Q_k'-iP_k')\\
n_k=a_k^\dagger a_k\ ,\qquad\textrm{ for }k=1,\ldots,n\ ,
\end{align*}
the state takes the form
\begin{align}
U_S\rho U_S^\dagger &=\bigotimes_{k=1}^n\frac{e^{-\beta_k n_k}}{\tr e^{-\beta_k n_k}}\ .\mylabel{eq:rhodiagonalized}
\end{align}

In this expression, the inverse temperatures $\beta_k$ are defined in terms of the symplectic eigenvalues as
\begin{align}
(e^{\beta_k}-1)^{-1}&=(\nu_k-1)/2:=N(\nu_k)\ \textrm{ for }k=1,\ldots, n\ .\mylabel{eq:meanphotonnumber}
\end{align}
The quantity~\eqref{eq:meanphotonnumber} is the {\em mean photon number}, i.e., the expectation value~$\tr(U_S\rho U_S^\dagger a_k^\dagger a_k)$. Since the number operators have the spectral decomposition~$n_k=\sum_{n\geq 0} n\proj{n}$ with integer eigenvalues, the entropy $S(\rho)=-\tr(\rho\log\rho)$  of a Gaussian state~$\rho$ can be evaluated from~\eqref{eq:rhodiagonalized} and~\eqref{eq:meanphotonnumber} as
\begin{align}
S(\rho)&=\sum_{k=1}^n g(N(\nu_k)) \mylabel{eq:entropygaussianexpr}
\end{align}
where $g(N):=(N+1)\log(N+1)-N\log N$. 
Note that this quantity does not depend on the displacement vector~$d$.

A completely positive trace-preserving map (CPTPM)~$\cE$ on $\cH$ is called {\em Gaussian} if it preserves the set of Gaussian states, that is, the state~$\rho':=\cE(\rho)$ is Gaussian for all Gaussian input states~$\rho$. By definition, the set of Gaussian operations is closed under composition. A Gaussian operation~$\cE$ is completely specified by its (Heisenberg) action on mode operators. The action of $\cE$ on Gaussian states is determined by  a triple $(X,Y,\xi)$, where $\xi\in\mathbb{R}^n$ is an arbitrary vector and $X,Y$ are real $2n\times 2n$~matrices satisfying $Y^T=Y$ and $Y+i(J-XJX^T)\geq 0$. For a Gaussian state~$\rho$ with covariance matrix~$\gamma$ and displacement vector~$d$, the Gaussian state~$\cE(\rho)$ is described by
\begin{align}
\gamma&\mapsto \gamma':=X\gamma X^T+Y\mylabel{eq:covariancematrixmap}\\
d&\mapsto d':=Xd+\xi\ . \mylabel{eq:displacementmap}
\end{align}
 More generally, the action on a general state $\rho$ is 
determined by the action
\begin{align}
\chi_\rho(\xi)&\mapsto \chi_{\cE(\rho)}(\xi)=\chi_\rho(X\xi)e^{-\frac{1}{4}\xi\cdot Y\xi}\ .\mylabel{eq:chxitransformationrule}
\end{align}
on characteristic functions.  In the special case where~$\cE(\rho)=U\rho U^\dagger$ is unitary, we have~$Y=0$ and $X=:S\in Sp(2n,\mathbb{R})$. We call such a unitary~$U$ Gaussian and sometimes write~$U=U_S$ to indicate it is defined by~$S$ (cf.~\eqref{eq:unitarymapweyl}).

A general Gaussian operation~$\cE$ on~$\cH_n$ has a Stinespring dilation with a Gaussian unitary~$U_S$, $S\in Sp(2(n+m),\mathbb{R})$ acting on the Hilbert space~$\cH_n\otimes\cH_m$ of $n+m$~modes (for some~$m$) and a Gaussian ancillary state~$\rho_B$ on~$\cH_m$, i.e., it is of the form
\begin{align}
\cE(\rho)&=\tr_B(U_S(\rho\otimes \rho_B)U_S^\dagger)\ .\mylabel{eq:stinespringdilation}
\end{align} 
Since the operations of taking the tensor product and the partial trace correspond to taking the direct sum or a submatrix on the level of covariance matrices, identities~\eqref{eq:covariancematrixmap} and~\eqref{eq:stinespringdilation} translate into
\begin{align*}
\left[S(\gamma\otimes \gamma_B)S^T\right]_{2n}&=X\gamma X^T+Y\ ,
\end{align*}
where $[\cdot]_{2n}$ denotes the leading principle $2n\times 2n$~submatrix and~$\gamma_B$ is the covariance matrix of~$\rho_B$. It is again convenient to express this in terms of characteristic functions. First consider the partial trace:
if $\rho_{n+m}$ has characteristic function $\chi_{\rho_{n+m}}(\xi,\xi')$, where $\xi\in\mathbb{R}^{2n}$, $\xi'\in\mathbb{R}^{2m}$, then the partial trace~$\rho_n=\tr_m\rho_{n+m}$ has characteristic function
\begin{align*}
\chi_{\rho_n}(\xi)&=\chi_{\rho_{n+m}}(\xi,0^{2m})\ .
\end{align*}
With~\eqref{eq:chxitransformationrule}, we get the transformation rule 
\begin{align}
\chi_\rho\mapsto \chi_{\cE(\rho)}(\xi):=(\chi\otimes\chi_B)(S(\xi,0^{2m}))e^{-\frac{1}{4}(\xi,0^{2m})\cdot Y(\xi,0^{2m})}\ ,\mylabel{eq:composedmapcharacteristic}
\end{align}
where $\chi_B$ is the  characteristic function of the Gaussian ancillary state~$\rho_B$.
Here $(\chi\otimes\chi_B)(\xi,\xi'):=\chi(\xi)\cdot\chi_B(\xi')$ is the characteristic function of the product state~$\rho\otimes\rho_B$.

\subsection{Quantum addition using the beamsplitter: definition\mylabel{sec:beamsplitterdef}}
In this section, we specify the quantum addition operation~\eqref{eq:additionlawquantumcov} in more detail. It takes takes two $n$-mode states~$\rho_X,\rho_Y$ and outputs an $n$-mode state $\rho_{X\boxplus_\lambda Y}=\cE_\lambda(\rho_X\otimes\rho_Y)$, where $0<\lambda<1$.

To define the CPTP map $\cE_\lambda$, consider the transmissivity $\lambda$-beam splitter. This is a Gaussian unitary~$U_{\lambda,n}$, whose action on $2n$~modes is defined by the symplectic matrix
\begin{align}
S_{\lambda,n}&=\left(
\begin{matrix}
\sqrt{\lambda}I_n & \sqrt{1-\lambda}I_n\\
\sqrt{1-\lambda}I_n & -\sqrt{\lambda}I_n
\end{matrix}
\right)\otimes I_2\ ,\label{eq:symplectictransfmatrix}
\end{align}
where we use a tensor product $\{e_1^X,\ldots,e_n^X,e_1^Y,\ldots,e_1^Y\}\otimes \{q,p\}$
to represent the two sets of modes~$(q^X_1,p^X_1,\ldots,q^X_n,p^X_n)$ and $(q^Y_1,p^Y_1,\ldots,q^Y_n,p^Y_n)$. Note that the symplectic matrix takes the form $I_n\otimes J$, with
\begin{align*}
J&=\left(
\begin{matrix}
0 & 1\\
-1 & 0
\end{matrix}
\right)\ .
\end{align*}
With respect to the $n$ pairs~$(q_j^X,p_j^X),(q_j^Y,q_j^Y)$ of modes, for $j=1,\ldots,n$, we have $S_{\lambda,n}=S_{\lambda,1}^{\oplus n}$, hence $U_{\lambda,n}=U_{\lambda,1}^{\otimes n}$ corresponds to  independently applying a beam-splitter to each pair of modes. In the following, we will omit the subscript~$n$, and simply write $\xi_X=(\xi_X^{Q_1},\xi_X^{P_1},\ldots,\xi_X^{Q_n},\xi_X^{P_n})$ and $\xi_Y=(\xi_Y^{Q_1},\xi_Y^{P_1},\ldots,\xi_Y^{Q_n},\xi_Y^{P_n})$ for the collection of phase space variables associated with first and second set of $n$~modes. 

The map $\cE_\lambda$ is defined by conjugation with $U_\lambda$ and tracing out the second set of modes, i.e., it is
\begin{align}
\cE_\lambda(\rho_{XY}) &=\tr_{Y} U_\lambda \rho_{XY} U_\lambda^\dagger\ .\mylabel{eq:elambdamapdefinition}
\end{align}
Clearly, this is a Gaussian map. A bipartite Gaussian state~$\rho_{XY}$ with covariance matrix and displacement vector
\begin{align}
\gamma&=
\left(\begin{matrix}
\gamma_{X} &\gamma_{XY}\\
\gamma_{YX} &\gamma_{Y}\\
\end{matrix}\right)\ ,\qquad d=(d_X,d_Y)\label{eq:bipartitegaussianstatedescr}
\end{align}
gets mapped into a Gaussian state $\cE_\lambda (\rho_{XY})$ with
\begin{align}
\begin{matrix}
\gamma'&=&\lambda \gamma_X+(1-\lambda)\gamma_Y+\sqrt{\lambda(1-\lambda)}(\gamma_{XY}+\gamma_{YX})\\
d'&=&\sqrt{\lambda}d_X+\sqrt{1-\lambda}d_Y
\end{matrix}\label{eq:covariancematrixadditionbeam}
\end{align}
This completely determines the action of $\cE_\lambda$ on Gaussian inputs, but we will also be interested in more
general inputs of product form. To get an explicit expression, consider two states $\rho_X,\rho_Y$ with characteristic functions~$\chi_X,\chi_Y$. According to~\eqref{eq:chxitransformationrule} and~\eqref{eq:symplectictransfmatrix}, the state~$U_\lambda(\rho_X\otimes\rho_Y)U^\dagger$  has characteristic function
\begin{align}
\chi_{U_\lambda(\rho_X\otimes\rho_Y)U^\dagger}(\xi_X,\xi_Y)=\chi_X(\sqrt{\lambda}\xi_X+\sqrt{1-\lambda}\xi_Y)\nonumber\\
\qquad\qquad\cdot\chi_Y(\sqrt{1-\lambda}\xi_X-\sqrt{\lambda}\xi_Y)\ .\mylabel{eq:characteristiculambda}
\end{align}
It follows that $\cE_\lambda(\rho\otimes\rho)$ has characteristic function
\begin{align}
\chi_{\cE_\lambda(\rho_X\otimes\rho_Y)}(\xi)&=\chi_X(\sqrt{\lambda}\xi)\cdot\chi_Y(\sqrt{1-\lambda}\xi)\ .
\mylabel{eq:characteristicfunctionatoutput}
\end{align}

\section{A quantum diffusion equation\label{sec:basicpropertiesliouville}}
Consider the diffusion Liouvillean
\begin{align}
\cL(\rho)&=- \frac{1}{4}\sum_{j=1}^n \left([Q_j,[Q_j,\cdot]]+[P_j,[P_j,\cdot]]\right)\nonumber\\
&=-\frac{1}{4}\sum_{j=1}^{2n}[R_j,[R_j,\rho]]\ .\label{eq:diffliou}
\end{align} 
defined on $n$ modes~(cf.~\eqref{eq:markovmaster}).   
We first establish the relevant properties of the one-parameter semigroup $\{e^{t\cL}\}_{t\geq 0}$, where the CPTP map~$e^{t\cL}$ describes evolution under the Markovian master equation
\begin{align} 
\frac{d}{dt}\rho(t)&=\cL(\rho(t))\ .\label{eq:markovianmastereq}
\end{align}
for time~$t$. We call this the diffusion semigroup. We will subsequently show that the maps $e^{t\cL}$ are compatible with beamsplitters~(Section~\ref{sec:beamsplittercompatibility}), and analyze the asymptotic scaling~$S(\rho(t))$ of  solutions $\rho(t)=e^{t\cL}(\rho)$ of~\eqref{eq:markovianmastereq} (Section~\ref{sec:scalingasymptotic}).

The following is well-known~(see e.g.,~\cite{Vanheuverzwijn1978} and~\cite{lloydetal09}, where~$\cL$ is given in terms of creation- and annihilation operators). 
\begin{lemma}\mylabel{lem:solution}
Let $\cL$ be the Liouvillean~\eqref{eq:diffliou}. Then
\begin{enumerate}[(i)]
\item
The Liouvillean~$\cL$ is Hermitian with respect to the Hilbert-Schmidt inner product, i.e.,
\begin{align}
\tr(\rho\cL(\sigma))=\tr(\cL(\rho)\sigma)\label{eq:Lhermitian}
\end{align}
for all states $\rho,\sigma$. 
\item
Let $\rho_0$ be a state  with characteristic function~$\chi_{\rho_0}(\xi)$, and let $\rho(t)=e^{t\cL}(\rho_0)$ denote the solution of the  Master equation~\eqref{eq:markovianmastereq} with initial value $\rho(0)=\rho_0$.
Then~$e^{t\cL}(\rho_0)$ has characteristic function
\begin{align}
\chi^{(t)}(\xi)&=\chi_{\rho_0}(\xi)\exp\left(-\|\xi\|^2t/4\right)\ ,\mylabel{eq:tentativechit}
\end{align}
where $\|\xi\|^2=\sum_{j=1}^{2n}\xi_j^2$.  
\item
For any $t\geq 0$, the CPTPM $e^{t\cL}$ is a Gaussian map acting on covariance matrices and displacement vectors by
\begin{align}
\begin{matrix}
\gamma&\mapsto \gamma'&=&\gamma+tI_{2n}\\
d&\mapsto d'&=&d\ .
\end{matrix}\label{eq:gaussianmapdescriptionsemigroup}
\end{align}
\end{enumerate}
\end{lemma}

\begin{proof}
If $\chi_\rho(\xi)$ is the characteristic function of a state~$\rho$, we
can express the commutator $[R_j,\rho]$ as
\begin{align*}
[R_j,\rho]&=\frac{1}{(2\pi)^n}\int\chi_\rho(\xi)\xi_jD(-\xi) d^2\xi\ .
\end{align*}
by using~\eqref{eq:rhocharacteristicfunction} and the displacement property~\eqref{eq:displacementproperty}. Iterating this argument gives
\begin{align*}
[R_j,[R_j,\rho]]&=\frac{1}{(2\pi)^n}\int\chi_\rho(\xi)\xi^2_jD(-\xi) d^2\xi\ ,
\end{align*}
hence
\begin{align}
\cL(\rho)&=-\frac{1}{4(2\pi)^n}\int\chi_\rho(\xi)\|\xi\|^2D(-\xi) d^2\xi\ .\mylabel{eq:Lappliedrho}
\end{align}
In other words, the operator $\cL(\rho)$ has characteristic function
$\chi_{\cL(\rho)}(\xi)=-\frac{1}{4}\chi_\rho(\xi)\|\xi\|^2$. 
Together with~\eqref{eq:overlapcomputationcharacteristic}, this immediately
implies~\eqref{eq:Lhermitian}.

Let $\rho(t)$ be the operator with characteristic function~\eqref{eq:tentativechit}. By commuting derivative and integration, we have
\begin{align*}
\frac{d}{dt}\rho(t)&=\frac{d}{dt}\frac{1}{(2\pi)^n}\int\chi_{\rho_0}(\xi)\exp(-\|\xi\|^2 t/4) D(-\xi)d^2\xi\\
&=-\frac{1}{4(2\pi)^n}\int\chi_{\rho_0}(\xi)\|\xi\|^2\exp(-\|\xi\|^2t/4)D(-\xi) d^2\xi\\
&=-\frac{1}{4(2\pi)^n}\int\chi_{\rho(t)}(\xi)\|\xi\|^2 D(-\xi)d^2\xi
\end{align*}
Combined with~\eqref{eq:Lappliedrho}, this shows that $\rho(t)$ is indeed the  solution to the Liouville equation, i.e.,~$\rho(t)=e^{t\cL}(\rho)$. This proves~$(ii)$. 

Finally, property~(iii) follows from~(ii) since multiplying a Gaussian characteristic function by a Gaussian preserves its Gaussianity.
 \end{proof}
 
\subsection{Compatibility with the beamsplitter\label{sec:beamsplittercompatibility}}
In the same way as convolution is compatible with the heat equation~(cf.~\eqref{eq:diffusingoutputequivdiffusinginput}), beamsplitters are compatible with diffusion defined by the Liouvillean~$\cL$. This can be understood as a consequence of the fact that, on the level of Wigner functions (i.e., the Fourier transform of the characteristic function), both the beam-splitter map~$\cE_\lambda$ and  the diffusion~$e^{t\cL}$ are described by a convolution.  Here we give a compact proof based on the fact that these maps are Gaussian.
\begin{lemma}\mylabel{lem:differentiation}
Let $t_X,t_Y\geq 0$, $0<\lambda<1$ and let $\cL$ be the diffusion Liouvillean acting on~$n$ modes. Then
\begin{align}
\cE_\lambda\circ(e^{t_X\cL}\otimes e^{t_Y\cL})\equiv e^{t\cL}\circ\cE_\lambda\label{eq:compatibilityone}
\end{align}
where $t=\lambda t_X+(1-\lambda)t_Y$.
\end{lemma} 
\begin{proof} 
According to Sections~\ref{sec:beamsplitterdef} and~\ref{sec:basicpropertiesliouville}, both~$\cE_\lambda$ and~$e^{t\cL}$ are Gaussian maps, hence this is true also for both sides of identity~\eqref{eq:compatibilityone}. This means it suffices to verify that they agree on all Gaussian input states~$\rho_{XY}$ described by~\eqref{eq:bipartitegaussianstatedescr}. The claim follows immediately from~\eqref{eq:covariancematrixadditionbeam} and~\eqref{eq:gaussianmapdescriptionsemigroup}: Indeed, both $\cE_{\lambda}(e^{t_X\cL}\otimes e^{t_Y\cL})(\rho_{XY})$ and $e^{t\cL}(\cE_\lambda(\rho_{XY}))$ are described by
\begin{align*}
\gamma'&=\lambda \gamma_X+(1-\lambda)\gamma_Y+tI\\
d'&=\sqrt{\lambda}d_X+\sqrt{1-\lambda}d_Y\ . 
\end{align*}
\end{proof}

\subsection{Scaling of the entropy as $t\rightarrow\infty$\label{sec:scalingasymptotic}}
The following result will be essential for the proof of the entropy power inequality. It derives from a combination of arguments from~\cite{Giovannettietal04} and~\cite{Hiroshima06}, as well as the maximum entropy principle~\cite{Wolfetal06}. 
\begin{theorem}[Entropy scaling under diffusion]\mylabel{thm:scaling}
Let $\rho$ be an arbitrary (not necessarily Gaussian) state, whose covariance matrix has symplectic eigenvalues~$\nu_1,\ldots,\nu_n$. Let $t>2$. Then 
\begin{align}
n g(N(t-1))&\leq S(e^{t\cL}(\rho))\leq \sum_{k=1}^n g(N(t+\nu_k))\ ,\mylabel{eq:scalingupperandlowerbound}
\end{align}
where $g(N)=(N+1)\log (N+1)-N\log N$ is the entropy of 
one-mode state with mean photon number~$N$ and $N(\nu)=(\nu-1)/2$.
\end{theorem}
\begin{proof}
The lower bound  is essentially the lower bound on the minimum output entropy from~\cite[(33)]{Giovannettietal04}, generalized to several modes. We recall the necessary definitions in the Appendix.  
Any state can be written in the form
\begin{align*}
\rho&=\int Q(\eta)\sigma(\eta)d^{2n}\eta\qquad\textrm{ with }\\
\qquad\sigma(\eta)&=\frac{1}{(2\pi)^n}\int e^{\|\xi\|^2/4}e^{i\xi\cdot J\eta} D(-\xi)  d^{2n}\xi\ ,
\end{align*}
where $Q$ is a probability distribution, i.e., $Q(\eta)\geq 0$ and $\int Q(\eta)d^{2n}\eta=1$. The operators~$\sigma(\eta)$ are generally not quantum states. However, by linearity of the superoperator~$e^{t\cL}$, we can compute
\begin{align*}
e^{t\cL}(\sigma(\eta))&=\frac{1}{(2\pi)^n}\int e^{-(t-1)/4\cdot\|\xi\|^2}e^{i\xi\cdot J\eta} D(-\xi)  d^{2n}\xi\
\end{align*}
according to Lemma~\ref{lem:solution}, which shows that $e^{t\cL}(\sigma(\eta))$ is a displaced thermal state with 
covariance matrix $(t-1)I_{2n}$ if $t>2$. 
Because entropy is concave, 
we therefore get
\begin{align*}
S(e^{t\cL}(\rho))&\geq \int S(e^{t\cL}(\sigma(\eta))) d^{2n}\eta=n g(N(t-1))\ ,
\end{align*}
and the claim follows.

For the proof of the upper bound in~\eqref{eq:scalingupperandlowerbound},
we can assume without loss of generality that~$\rho$ is Gaussian. This is because the maximum entropy principle~\cite{Wolfetal06} states  that for any state~$\rho$, we have
\begin{align}
S(\rho)\leq S(\rho^G)\ ,\label{eq:maximumentropyprinciplewrittenout}
\end{align}
where~$\rho^G$ is a Gaussian state with covariance matrix identical to~$\rho$.
The remainder of the proof  closely follows arguments in~\cite{Hiroshima06}.  It is a well-known fact that the entropy of a state $\rho$ can be computed by taking the limit
\begin{align}
S(\rho)&=\lim_{p\rightarrow 1^+}\frac{1}{1-p}\log \tr\rho^p\ .\mylabel{eq:entropypnormlimit}
\end{align}
For a Gaussian state~$\rho$ whose covariance matrix has symplectic eigenvalues~$\nu=(\nu_1,\ldots,\nu_n)$, the expression on the rhs. is equal to~\cite{holevosohmaetal}
\begin{align}
\tr \rho^p &=\prod_{k=1}^n \frac{2^p}{f_p(\nu_j)}:=\frac{2^{pn}}{F_p(\nu)}\qquad\textrm{ and where }\nonumber\\
f_p(\nu)&=(\nu+1)^p-(\nu-1)^p\ \mylabel{eq:fpdef}
\end{align}
for any $p\geq 1$.  Weak submajorization, written $x\prec^w y$ for two vectors $x,y\in\mathbb{R}^n$, is defined by the condition
\begin{align*}
x\prec^w y\qquad\Leftrightarrow\qquad\sum_{j=1}^k x^\downarrow_j\leq \sum_{j=1}^k y^\downarrow_j\qquad\textrm{ for }k=1,\ldots,n\ ,
\end{align*}
where $x^\downarrow_1\geq x^\downarrow_2\geq \cdots \geq x^\downarrow_n$ and $y^\downarrow_1\geq y^\downarrow_2\geq \cdots \geq y^\downarrow_n$ are decreasing rearrangements of $x$ and $y$, respectively. In~\cite{Hiroshima06}, it is argued that the function $\nu\mapsto F_p(\nu)$ respects the partial order imposed by weak submajorization on~$\mathbb{R}^n$, i.e., it has the property
\begin{align}
F_p(x)\geq F_p(y)\qquad\textrm{ if }\qquad x\prec^w y\ .\mylabel{eq:fpxyresp}
\end{align}
Furthermore,~\cite[Theorem~1]{Hiroshima06} states that for any $2n\times 2n$~real positive symmetric matrices~$A$ and~$B$, the vector of symplectic eigenvalues of their sum~$A+B$ is submajorized by the sum of the corresponding vectors of $A$ and $B$, respectively, i.e.,
\begin{align} 
\nu(A+B)\prec^w (\nu(A)+\nu(B))\ .\mylabel{eq:summariz}
\end{align}
Combining~\eqref{eq:entropypnormlimit},~\eqref{eq:fpdef},~\eqref{eq:fpxyresp} and~\eqref{eq:summariz} gives the following statement.
Let $\rho[\gamma]$ denote the centered Gaussian state with covariance matrix~$\gamma$. Let $\gamma_A>0$ and $\gamma_B>0$ be positive covariance matrices with symplectic eigenvalues~$\nu^A,\nu^B$. Then
\begin{align}
S(\rho[\gamma_A+\gamma_B])\leq S\left(\rho\left[\bigoplus_{j=1}^n (\nu^A_j+\nu^B_j)I_2\right]\right)\ .\mylabel{eq:covariancethermalinequality}
\end{align}
The upper bound~\eqref{eq:scalingupperandlowerbound} follows immediately from~\eqref{eq:maximumentropyprinciplewrittenout} and~\eqref{eq:covariancethermalinequality} applied to $\gamma_A=\gamma$ and $\gamma_B=tI_{2n}$ (a valid covariance matrix for $t\geq 1$). This is because for an initial state $\rho$ with covariance matrix~$\gamma$,  the time-evolved state $e^{t\cL}(\rho)$ has covariance matrix~$\gamma+tI_{2n}$ according to Lemma~\ref{lem:solution}. 
\end{proof}

Using Theorem~\ref{thm:scaling}, we can show that the asymptotic scaling of the entropy (or the entropy power) is independent of the initial state~$\rho$, in the following sense. 
\begin{corollary}\label{cor:scaling}
The entropy and the entropy power of $e^{t\cL}(\rho)$
 grow logarithmically respectively linearly as $t\rightarrow\infty$, with an asymptotic time-dependence independent of the  initial state~$\rho$.  More precisely, we have
\begin{align*}
\big|S(e^{t\cL}(\rho))/n-(1-\log 2+\log t)\big| &\leq O(1/t)\\
\big|\frac{1}{t}e^{S(e^{t\cL}(\rho))/n}-\frac{e}{2}\big|&\leq O(1/t)\ .
\end{align*}
where the constants in $O(\cdot)$ depend on~$\rho$.
\end{corollary}

\begin{proof}
From the Taylor series expansions
\begin{align*}
\log(1+\epsilon)&=\epsilon+O(\epsilon^2)\\
(1+\epsilon)^{1/\epsilon}&=e-\frac{e}{2}\epsilon+O(\epsilon^2)
\end{align*}
we get
\begin{align*}
\begin{matrix}
g(N)&=&\log N+1+O(1/N)\\ 
e^{g(N)}&=&\left(\frac{1}{2}+N\right)e+O(\frac{1}{N})\qquad \textrm{ for }N\rightarrow\infty\ .
\end{matrix}
\end{align*}
Replacing the rhs of~\eqref{eq:scalingupperandlowerbound} with the upper bound $\sum_{k=1}^n g(N(t+\nu_k))\leq n\max_k g(N(t+\nu_k))=:ng(N(t+\nu_*))$ and using $N(\nu)=(\nu-1)/2$ therefore gives 
\begin{align*}
\begin{split}
\log(\frac{t}{2}-1)+1-O(1/t)&\leq S(e^{t\cL}(\rho))/n\\
& \leq \log(\frac{t+\nu_*-1}{2})+1+O(1/t)
\end{split}
\end{align*}
and 
\begin{align*}
\begin{matrix}
\frac{t}{2}e-\frac{1}{2}e-O(\frac{1}{t})&\leq e^{S(e^{t\cL}(\rho))/n}&\leq& \frac{t}{2}e+ \frac{\nu_*}{2}e+O(\frac{1}{t})
\end{matrix}
\end{align*}
for $t\rightarrow\infty$, 
for some constant $\nu_*$ depending on $\rho$. The claim follows from this. 
\end{proof}

A simple consequence of Corollary~\ref{cor:scaling} is that $e^{t\cL}(\rho)$ converges in relative entropy to a Gaussian state: we have
\begin{align}
\lim_{t\rightarrow \infty}S(e^{t\cL}(\rho)\|e^{t\cL}(\rho)^G)=0\ ,\mylabel{eq:limitdistancetogaussian}
\end{align}
where $\sigma^G$ denotes the `Gaussification' of $\sigma$, i.e., the Gaussian state with identical first and second moments. Indeed, because
$e^{t\cL}$ is a Gaussian map (Lemma~\ref{lem:solution}), we have $e^{t\cL}(\rho)^G=e^{t\cL}(\rho^G)$ for any state $\rho$. Therefore,
\begin{align*}
S(e^{t\cL}(\rho)\|e^{t\cL}(\rho)^G)&=S(e^{t\cL}(\rho)\|e^{t\cL}(\rho^G))\\
&=-S(e^{t\cL}(\rho))-\tr(e^{t\cL}(\rho)\log e^{t\cL}(\rho^G))\\
&=-S(e^{t\cL}(\rho))-\tr(e^{t\cL}(\rho)^G\log e^{t\cL}(\rho^G))\\
&=-S(e^{t\cL}(\rho))-\tr(e^{t\cL}(\rho^G)\log e^{t\cL}(\rho^G))\\
&=-S(e^{t\cL}(\rho))+S(e^{t\cL}(\rho^G))\ .
\end{align*}
In the third identity, we used the fact that $\log e^{t\cL}(\rho^G)$ is quadratic in the mode operators. Statement~\eqref{eq:limitdistancetogaussian} now follows from Corollary~\ref{cor:scaling}. 
 \section{Divergence-based Quantum Fisher
information\label{sec:divergencebasedqfisher}}
In this section, we introduce the divergence-based Fisher information
and establish its main properties. We restrict our focus to what is needed for 
the proof of the entropy power inequality. We refer to the literature for substantially more general results concerning quasi-entropies and quantum Fisher information (see e.g.,~\cite{Petz96,LesRus99,petzghi10,reviewhiai11}).

Recall that the relative entropy or divergence of two (invertible) density operators~$\rho$,~$\sigma$ is defined as $S(\rho\|\sigma)=\tr\rho(\log\rho-\log\sigma)$. The divergence is nonnegative and faithful, i.e., 
\begin{align}
S(\rho\|\sigma)\geq 0\ \textrm{ and }\  S(\rho\|\sigma)=0\ \ \textrm{ if and only if }\ \rho=\sigma\ .\label{eq:relativeentropynonneg}
\end{align}
Furthermore, it is monotone, i.e., non-increasing under the application of a CPTPM~$\cE$,
\begin{align}
S(\cE(\rho)\|\cE(\sigma))\leq S(\rho\|\sigma)\ .\label{eq:dataprocessingdivergence} 
\end{align}
Properties~\eqref{eq:relativeentropynonneg} and~\eqref{eq:dataprocessingdivergence} tell us that $S(\cdot\|\cdot)$ may be thought of as a measure of closeness.
 A third important property of divergence is its additivity under tensor products,
\begin{align}
S(\rho_1\otimes\rho_2\|\sigma_1\otimes\sigma_2)=S(\rho_1\|\sigma_1)+S(\rho_2\|\sigma_2)\ .\label{eq:divergenceadditivity}
\end{align}

Consider a smooth one-parameter family $\theta\mapsto \rho^{(\theta)}$ of states. Using divergence to quantify how much
these states change as we deviate from a basepoint $\theta_0\in\mathbb{R}$, it is natural to consider the function $\theta\mapsto S(\rho^{(\theta_0)}\|\rho^{(\theta)})$ in the neighborhood of~$\theta_0$.  In the following, we assume that it is twice differentiable at~$\theta_0$ (but see comment below). According to~\eqref{eq:relativeentropynonneg}, this function vanishes for $\theta=\theta_0$, and is nonnegative everywhere. This implies that its first derivative vanishes, i.e., 
\begin{align}
\frac{d}{d\theta}S(\rho^{(\theta_0)}\|\rho^{(\theta)})\Big|_{\theta=\theta_0}=0\ .\label{eq:permissibilitycondition}
\end{align}
One is therefore led to consider the second derivative, which we will denote by
\begin{align}
J(\rho^{(\theta)};\theta)|_{\theta=\theta_0}&=\frac{d^2}{d\theta^2}S(\rho^{(\theta_0)}\|\rho^{(\theta)})\Big|_{\theta=\theta_0}\ .\label{eq:divergencebasedfisherinformationdef}
\end{align}
We call the quantity~\eqref{eq:divergencebasedfisherinformationdef} the divergence-based Fisher information of the family~$\{\rho^{(\theta)}\}_\theta$. Its properties are as follows:

\begin{lemma}[Reparametrization formula]
For any constant $c$, we have
\begin{align*}
J(\rho^{(c\theta)};\theta)|_{\theta=\theta_0}&=c^2 J(\rho^{(\theta)};\theta)|_{\theta=\theta_0}\\
J(\rho^{(\theta+c)};\theta)|_{\theta=\theta_0}&= J(\rho^{(\theta)};\theta)|_{\theta=\theta_0+c}\ .
\end{align*}
\end{lemma}
\begin{proof}
This is immediate from Definition~\eqref{eq:divergencebasedfisherinformationdef}.
\end{proof}

\begin{lemma}[Non-negativity]\mylabel{lem:fisherinformationpositive}
The Fisher information satisfies
\begin{align}
J(\rho^{(\theta)}; \theta)\big|_{\theta=\theta_0} \geq 0.
\end{align}
\end{lemma}
\begin{proof}
Define $f_{\theta_0}(\theta)=S(\rho^{(\theta_0)}\| \rho^{(\theta)})$. We know that
\begin{align*}
f_{\theta_0}(\theta)\geq 0\ \textrm{ for all }\theta\qquad\textrm{ and }\qquad f_{\theta_0}(\theta_0)=0\ .
\end{align*}
The claim follows from this by writing the second derivative as a limit
\begin{align}
\begin{split}
\frac{d^2}{d\theta^2}S(\rho^{(\theta_0)}\| \rho^{(\theta)})\Big|_{\theta=\theta_0}
&\vspace{-1ex}=\frac{d^2}{d\theta^2}f_{\theta_0}(\theta)\Big|_{\theta=\theta_0}\\
=\lim_{\epsilon\rightarrow 0}&\frac{f_{\theta_0}(\theta_0+\epsilon)-2f_{\theta_0}(\theta_0)+f_{\theta_0}(\theta_0-\epsilon)}{\epsilon^2}\ .
\end{split}\label{eq:secondderivativeasalimit}
\end{align}
\end{proof}

\begin{lemma}[Additivity]
For a family~$\{\rho_A^{(\theta)}\otimes\rho_B^{(\theta)}\}_\theta$ of bipartite product states, we have
\begin{align}
J(\rho_A^{(\theta)}\otimes\rho_B^{(\theta)};\theta)\Big|_{\theta=\theta_0}&=J(\rho_A^{(\theta)};\theta)\Big|_{\theta=\theta_0}+J(\rho_B^{(\theta)};\theta)\Big|_{\theta=\theta_0}\
.\mylabel{eq:fisherinformatioadditivity}
\end{align}
\end{lemma}
\begin{proof}
This directly follows from the additivity~\eqref{eq:divergenceadditivity} of divergence. 
\end{proof}

Most importantly, the divergence-based Fisher information satisfies the data processing inequality.
\begin{theorem}[Data processing for Fisher
information]\mylabel{thm:dataprocessing}
Let $\cE$ be an arbitrary CPTPM. Then
\begin{align*}
J(\cE(\rho^{(\theta)});\theta)\big|_{\theta=\theta_0}\leq
J(\rho^{(\theta)};\theta)\big|_{\theta=\theta_0}\ .
\end{align*}
\end{theorem}
\begin{proof}
The proof proceeds in the same way as a well-known proof for the
information-processing property of classical Fisher information (see e.g.,~\cite{Rag11}). We show here that our assumption that~$f_{\theta_0}(\theta)=S(\rho^{(\theta_0)}\|\rho^{(\theta)})$ is twice differentiable at $\theta=\theta_0$ is sufficient to give the desired inequality. 
 From~\eqref{eq:permissibilitycondition} and
the definition of Fisher information, we conclude that
\begin{align}
S(\rho^{(\theta_0)}\|\rho^{(\theta)})&=\frac{1}{2}J(\rho^{(\theta)};\theta)\Big|_{\theta=\theta_0}\cdot
(\theta-\theta_0)^2+o(|\theta-\theta_0|^2)\ .\label{eq:fishexpone}
\end{align}
Let us argue that an analogous identity holds for the family~$\{\cE(\rho^{(\theta)})\}_\theta$.   Define $g_{\theta_0}(\theta)=S(\cE(\rho^{\theta_0})\|\cE(\rho^{\theta}))$. We know from~\eqref{eq:relativeentropynonneg}, \eqref{eq:permissibilitycondition} and the monotonicity~\eqref{eq:dataprocessingdivergence} that
\begin{align}
\begin{split}
0\leq g_{\theta_0}(\theta)\leq f_{\theta_0}(\theta)\ \textrm{ for all }\theta\textrm{ and }\\
g_{\theta_0}(\theta_0)=f_{\theta_0}(\theta_0)=\frac{d}{d\theta}f_{\theta_0}(\theta)\Big|_{\theta=\theta_0}=0\ .
\end{split}\label{eq:nonnegativeexprone}
\end{align}
Therefore, the derivative of $g_{\theta_0}$ is
\begin{align*}
\frac{d}{d\theta}g_{\theta_0}(\theta)\big|_{\theta=\theta_0}&=\lim_{\epsilon\rightarrow 0}\frac{g_{\theta_0}(\theta_0+\epsilon)-g_{\theta_0}(\theta_0)}{\epsilon}\\
&=\lim_{\epsilon\rightarrow 0}\frac{g_{\theta_0}(\theta_0+\epsilon)}{\epsilon}\geq 0\ .
\end{align*}
On the other hand,~\eqref{eq:nonnegativeexprone}
also gives 
\begin{align*}
\frac{d}{d\theta}g_{\theta_0}(\theta)\big|_{\theta=\theta_0}&=\lim_{\epsilon\rightarrow 0}\frac{g_{\theta_0}(\theta_0+\epsilon)}{\epsilon}\leq\lim_{\epsilon\rightarrow 0}\frac{f_{\theta_0}(\theta_0+\epsilon)}{\epsilon}\\
&=\frac{d}{d\theta}f_{\theta_0}(\theta)\Big|_{\theta=\theta_0}=0\ ,
\end{align*}
hence
\begin{align*}
\frac{d}{d\theta}g_{\theta_0}(\theta)\Big|_{\theta=\theta_0}=0\ .
\end{align*}
Similarly, writing the second derivative of $g_{\theta_0}$ as a limit as in~\eqref{eq:secondderivativeasalimit} and using~\eqref{eq:nonnegativeexprone} gives
\begin{align*}
0\leq \frac{d^2}{d\theta^2}g_{\theta_0}(\theta)\Big|_{\theta=\theta_0}\leq \frac{d^2}{d\theta^2}f_{\theta_0}(\theta)\Big|_{\theta=\theta_0}\ .
\end{align*}
In summary, the function $g_{\theta_0}$ has vanishing first and bounded second derivative, which shows that
\begin{align}
S(\cE(\rho^{(\theta_0)})\|\cE(\rho^{(\theta)}))&=\frac{1}{2}J(\cE(\rho^{(\theta)});\theta)\Big|_{\theta=\theta_0}\hspace{-1ex}\cdot
(\theta-\theta_0)^2+o(|\theta-\theta_0|^2)\label{eq:fishexptwo}\
\end{align} 
The claim of the theorem now follows from~\eqref{eq:fishexpone},~\eqref{eq:fishexptwo}  and the data processing
inequality~\eqref{eq:dataprocessingdivergence}. 
\end{proof}
One may worry about the differentiability of the function  $\theta\mapsto S(\rho_{\theta_0}\|\rho_{\theta})$ for an arbitrary smooth one-parameter family~$\{\rho^{(\theta)}\}_\theta$ on an infinite-dimensional Hilbert space. However, we do not need this full generality here. Throughout, we are only concerned with covariant families of states 
\begin{align}
\rho^{(\theta)}=e^{i(\theta-\theta_0) H}\rho^{(\theta_0)}e^{-i(\theta-\theta_0)H}\qquad\textrm{ for all }\theta,\theta_0\in\mathbb{R}\label{eq:covariantfamily}
\end{align}
generated by a Hamiltonian~$H$. For a family of the form~\eqref{eq:covariantfamily}, we can easily compute the derivative
\begin{align}
\frac{d}{d\theta}S(\rho^{(\theta_0)}\|\rho^{(\theta)})\Big|_{\theta=\theta_0}&=-\tr\left(\rho^{(\theta_0)}[iH,\log\rho^{(\theta_0)}]\right)\ ,
\end{align}
In accordance with~\eqref{eq:permissibilitycondition}, this can be seen to vanish by inserting the spectral decomposition of~$\rho^{(\theta_0)}$.  We can also give an explicit expression for the second derivative, the divergence-based Fisher information. It will be convenient to state this as a lemma.
\begin{lemma}\mylabel{lem:fisherinformationcomputation}
Let $\{\rho^{(\theta)}\}_\theta$ be a covariant family of states as in~\eqref{eq:covariantfamily}. Then
\begin{align}
J(\rho^{(\theta)};\theta)|_{\theta=\theta_0}&=\tr(\rho^{(\theta_0)}[H,[H,\log\rho^{(\theta_0)}]])\
.\label{eq:explicitfisherinfocovariant}
\end{align}
\end{lemma}
\begin{proof}
Define $U(\theta) = e^{i(\theta-\theta_0) H}$.  Using $U(\theta)H U(\theta)^\dagger=H$, we calculate
\begin{align*}
\frac{d}{d\theta}\log\rho^{(\theta)}&=[iH,U(\theta)(\log\rho^{(\theta_0)})U(\theta)^\dagger]\\
&=U(\theta)[iH,\log\rho^{(\theta_0)}]U(\theta)^\dagger\ ,
\end{align*}
hence
\begin{align*}
\frac{d^2}{d\theta^2}\log\rho^{(\theta)}&=[iH,U(\theta)[iH,\log\rho^{(\theta_0)}]U(\theta)^\dagger]\
.
\end{align*}
Therefore
\begin{align*}
\frac{d^2}{d\theta^2} S(\rho^{(\theta_0)}\|\rho^{(\theta)})&=
-\tr(\rho^{(\theta_0)}
\frac{d^2}{d\theta^2}
\log\rho^{(\theta)})\\
&=-\tr(\rho^{(\theta_0)}\left[iH,U(\theta)[iH,\log\rho^{(\theta_0)}]U(\theta)^\dagger\right])\
\end{align*}
and evaluating this at $\theta=\theta_0$
gives the claim.
\end{proof}

\section{The quantum de Bruijin identity}
The quantum analog of the translation rule~\eqref{eq:translationsbytheta} is defined in terms of phase space translations: For $R\in\{Q_j,P_j\}$ define the displacement operator in the direction $R$ as
\begin{align}
D_R(\theta)=\begin{cases}
e^{i\theta P_j}\qquad &\textrm{ if }R=Q_j\\
e^{-i\theta Q_j}\qquad &\textrm{ if }R=P_j\ .
\end{cases}\mylabel{eq:translationsdefinition}
\end{align}
For a state $\rho$, consider the family of translated states
\begin{align}
\rho^{(\theta,R)}&=D_{R}(\theta)\rho D_{R}(\theta)^\dagger\qquad \theta\in\mathbb{R}\ \label{eq:displacedstatedef}
\end{align}
and its Fisher information $J(\rho^{(\theta,R)};\theta)\big|_{\theta=0}$. With a slight abuse of terminology, we will call the quantity
\begin{align}
J(\rho)&:=\sum_{k=1}^{2n}J(\rho^{(\theta,R_k)};\theta)\big|_{\theta=0}\ \mylabel{eq:fisherinformationsummed}
\end{align}
obtained by summing over all phase space directions 
the {\em Fisher information} of~$\rho$. The quantum version of de Bruijin's identity~\eqref{eq:debruijiclassicalstated} then reads as follows:

\begin{theorem}[Quantum de Bruijin]\label{thm:debruijin}
Let $\cL$ be the Liouvillean~\eqref{eq:diffliou}.
 The rate of entropy increase when evolving under the one-parameter semigroup
$\{e^{t\cL}\}_{t\geq 0}$ from the initial state~$\rho$ is given by
\begin{align}
\frac{d}{dt}S(e^{\cL t}(\rho))\big|_{t=0}&=\frac{1}{4}J(\rho)\  .\mylabel{eq:computedrateentropychange}
\end{align}
\end{theorem} 
\begin{proof}
We  use the expression~(see e.g.,\cite{spohn78})
\begin{align}
\frac{d}{dt}S(e^{t\cL}(\rho))\big|_{t=0}&=-\tr(\cL(\rho)\log\rho)\ .
\end{align}
for the rate of entropy increase under a one-parameter semigroup~$e^{t\cL}$.
Because $\cL$ is Hermitian (Lemma~\ref{lem:solution}), this is equal to
\begin{align}
\frac{d}{dt}S(e^{t\cL}(\rho))\big|_{t=0}&=-\tr(\rho\cL(\log\rho))\ ,\mylabel{eq:firstderivativeentropychange}
\end{align}
 The claim now follows from the definition~\eqref{eq:diffliou}
of the Liouvillean~$\cL$ and Lemma~\ref{lem:fisherinformationcomputation}.
\end{proof}
 
\section{Fisher information inequalities for the beamsplitter}
\subsection{Compatibility of the beam splitter with translations}
 We first establish a slight generalization of the compatibility condition~\eqref{eq:additiontranslationconv} between the addition rule defined by the CPTPM~$\cE_\lambda$ (cf.~\eqref{eq:elambdamapdefinition}) and translations~$D_R$ in phase space~(cf.~\eqref{eq:translationsdefinition}).
\begin{lemma}\label{eq:displacementlemma}
Let $\theta,w_X,w_Y\in\mathbb{R}$ and consider the CPTP maps
\begin{align*}
 \cF(\rho)=\qquad\qquad\hspace{40ex}\\
\  \cE_\lambda\left(\left(D_R(w_X\theta)\otimes D_R(w_Y\theta)\right)
\rho\left(D_R(w_X\theta)\otimes D_R(w_Y\theta)\right)^\dagger\right)
\end{align*}
and
\begin{align*}
\cG(\rho)&=D_R(w\theta)\cE_\lambda(\rho)D_R(w\theta)^\dagger
\end{align*}
where  $w =\sqrt{\lambda}w_X+\sqrt{1-\lambda}w_Y$.
Then $\cF\equiv \cG$. In particular, 
\begin{align*}
\cE_\lambda(\rho_{X}^{(w_X\theta,R)}\otimes\rho_{Y}^{(w_Y\theta,R)}) &=\left(\cE_\lambda(\rho_X\otimes \rho_Y)\right)^{(w\theta,R)}
\end{align*}
for any two $n$-mode states $\rho_X$ and $\rho_Y$, where $\rho\mapsto \rho^{(\theta,R)}$ is defined by~\eqref{eq:displacedstatedef}.
\end{lemma}
\begin{proof}
  Since both $\cF$ and $\cG$ are Gaussian (as compositions of Gaussian operations), it suffices to show that they agree on Gaussian inputs. Note that $\rho\mapsto D_R(\theta)\rho D_R(\theta)^\dagger$ leaves the covariance matrix $\rho$ invariant while adding $d'=d+\theta\cdot e_{R}$  to the displacement vector, where~$e_{R}$ is the standard basis vector in the $R$-direction. This, together with~\eqref{eq:covariancematrixadditionbeam}, immediately implies that~$\cF\equiv \cG$.
\end{proof}

\subsection{Stam inequality and convexity inequality for Fisher information}
The following Fisher information inequality for the quantity~$J(\rho)$ (cf.~\eqref{eq:fisherinformationsummed})  is a straightforward consequence of the compatibility of beam-splitters with translations, and the data processing inequality for Fisher information. The arguments used here are identical to those applied in Zamir's proof~\cite{Zamir98} (described transparently in~\cite{Kaganetal} for the special case of interest here) of Stam's inequality. The following statement implies a quantum version of the latter and a convexity inequality of the form~\eqref{eq:fisherinformationinequalitycomplete}.
\begin{theorem}[Quantum Fisher information inequality]\label{thm:qfishconv}
Let $w_X,w_Y\in\mathbb{R}$ and $0<\lambda<1$.  Let $\rho_X,\rho_Y$ be two $n$-mode states.  
Then
\begin{align*}
w^2J(\cE_\lambda(\rho_X\otimes\rho_Y))&\leq w_X^2 J(\rho_X)+w_Y^2J(\rho_Y)
\end{align*}
where $w=\sqrt{\lambda}w_X+\sqrt{1-\lambda}w_Y$.
\end{theorem}
\begin{proof}
Let $R\in \{Q_j,P_j\}_{j=1}^n$ and consider the family of product states~$\{\rho_{X}^{(w_X\theta,R)}\otimes\rho_{Y}^{(w_Y\theta,R)}\}_\theta$ and the family~$\{\left(\cE_\lambda(\rho_X\otimes\rho_Y)\right)^{(w\theta,R)}\}_\theta$. By Lemma~\ref{eq:displacementlemma}, the latter is obtained by data processing from the former. By the data processing inequality for Fisher information, this implies
\begin{align*}
J\left(
\left(\cE_\lambda(\rho_X\otimes\rho_Y)\right)^{(w\theta,R)}\right)\Big|_{\theta=0}
&\hspace{-1ex}\leq J(\rho_X^{(w_X\theta,R)}\otimes \rho_Y^{(w_Y\theta,R)};\theta)\Big|_{\theta=0}\ ,
\end{align*}
and by the properties of Fisher information, this is equivalent to
\begin{align*}
\begin{split}
w^2 J\left(
\left(\cE_\lambda(\rho_X\otimes\rho_Y)\right)^{(w\theta,R)}\right)\Big|_{\theta=0}\hspace{30ex}\\
\ \leq w_X^2 J(\rho_X^{(\theta,R)})|_{\theta=0}+w_Y^2J(\rho_Y^{(\theta,R)};\theta)\Big|_{\theta=0}\ .
\end{split}
\end{align*}
The claim is obtained by summing over all phase space directions~$R$. 
\end{proof}
The following convexity inequality for Fisher information is an immediate consequence of Theorem~\ref{thm:qfishconv}, obtained by setting $w_X=\sqrt{\lambda}$ and $w_Y=\sqrt{1-\lambda}$.
\begin{corollary}[Convexity inequality for Fisher information]\mylabel{cor:convexitypartial}
Let $\rho_X,\rho_Y$ be $n$-mode states and  $0<\lambda<1$. Then
\begin{align}
J(\cE_{\lambda}(\rho_X\otimes\rho_Y))&\leq \lambda J(\rho_X)+(1-\lambda)J(\rho_Y)\ .
\end{align}
\end{corollary} 
We also obtain the following quantum analog of the classical Stam inequality~\eqref{eq:staminequalityclassical}.

\begin{corollary}[Quantum Stam inequality]\mylabel{cor:staminequality}
Let $\rho_X,\rho_Y$ be arbitrary states and consider the $50:50$-beamsplitter map $\cE=\cE_{1/2}$. Then
\begin{align*}
\frac{2}{J(\cE(\rho_X\otimes\rho_Y))}\geq \frac{1}{J(\rho_X)}+\frac{1}{J(\rho_Y)}\ .
\end{align*}
\end{corollary}
\begin{proof}
This follows from Theorem~\ref{thm:qfishconv} by setting $(\lambda=1/2)$ and 
\begin{align*}
w_X &=\frac{J(\rho_X)^{-1}}{J(\rho_X)^{-1}+J(\rho_Y)^{-1}}\qquad\textrm{ and }\\
w_Y &=\frac{J(\rho_Y)^{-1}}{J(\rho_X)^{-1}+J(\rho_Y)^{-1}}\ .
\end{align*}
\end{proof}

 \section{Quantum entropy power inequalities}
Having introduced the right quantum counterparts of the relevant classical concepts, the proof of the quantum entropy power inequalities is straightforward. We first provide a proof of
the version~\eqref{eq:liebversionquantum} for the transmissivity~$\lambda$-beamsplitter. It closely follows an argument given by Dembo, Cover and Thomas in~\cite{Demboetal91} for the classical statement~\eqref{eq:liebversion}, but relies on a slight adaptation that allows us to generalize it to the quantum setting.
\begin{theorem}[Quantum entropy power inequality for transmissivity~$\lambda$]
Let $0<\lambda<1$ and let $\cE_\lambda$ be the map~\eqref{eq:elambdamapdefinition} associated with a  beamsplitter of transmissivity~$\lambda$. Let~$\rho_X$ and~$\rho_Y$ be arbitrary $n$-mode states. Then \mylabel{thm:main}
\begin{align}
S(\cE_\lambda(\rho_X\otimes\rho_Y))&\geq \lambda S(\rho_X)+(1-\lambda) S(\rho_Y)\ .\mylabel{eq:entropypowerinequality}
\end{align}
\end{theorem}

\begin{proof}
Define
\begin{align*}
\begin{split}
\delta(t)&=S(e^{t\cL}(\cE_\lambda(\rho_X\otimes\rho_Y)))-\lambda S(e^{t\cL}(\rho_X))\\
&\qquad\qquad\qquad\qquad \qquad -(1-\lambda) S(e^{t\cL}(\rho_Y))\  .
\end{split}
\end{align*}
We want to show that $\delta(0)\geq 0$. Since $\lim_{t\rightarrow\infty}\delta(t)=0$ by the asymptotic scaling of these entropies (Corollary~\ref{cor:scaling}), it suffices to show that 
\begin{align}
\frac{d}{dt}\delta(t)\leq 0\qquad\textrm{ for all }t\geq 0\ .\mylabel{eq:Ssmallzeroeq}
\end{align}
Because $\{e^{t\cL}\}_t$ is a semigroup, we get
\begin{align*}
\begin{split}
4\frac{d}{dt}\delta(t)&=J(e^{t\cL}(\cE_\lambda(\rho_X\otimes\rho_Y)))-\lambda J(e^{t\cL}(\rho_X))\\
&\qquad\qquad\qquad\qquad -(1-\lambda) J(e^{t\cL}(\rho_Y))\ ,
\end{split}
\end{align*}    
by the quantum de Bruijin identity~(Theorem~\ref{thm:debruijin}). Because of the compatibility of $\cE_\lambda$ with $e^{t\cL}$ (Lemma~\ref{lem:differentiation}), we can rewrite this as
\begin{align*}
\begin{split}
4\frac{d}{dt}\delta(t)&=J(\cE_\lambda(e^{t\cL}(\rho_X)\otimes e^{t\cL}(\rho_Y)))-\lambda J(e^{t\cL}(\rho_X))\\
&\qquad\qquad\qquad\qquad\qquad-(1-\lambda) J(e^{t\cL}(\rho_Y))\ .
\end{split}
\end{align*}  
Using the convexity inequality for Fisher information (Corollary~\ref{cor:convexitypartial}), the inequality~\eqref{eq:Ssmallzeroeq} follows.
\end{proof}

Specializing Theorem~\ref{thm:main} to  Gaussian states, we obtain a concativity inequality for the entropy with respect to the covariance matrix. That is, denoting by $\rho[\gamma]$ the centered Gaussian state with covariance matrix~$\gamma$, we have
\begin{align}
S(\lambda\rho[\gamma_1]+(1-\lambda)\rho[\gamma_2])\geq \lambda S(\rho[\gamma_1])+(1-\lambda)S(\rho[\gamma_2])
\end{align}
for all $0<\lambda<1$.

Finally, we can obtain the quantum analog~\eqref{eq:epiquantumcov} of the classical entropy power inequality~\eqref{eq:epiclassicalcov} for the $50:50$-beamsplitter, again mimicing a known classical proof.

\begin{theorem}[Entropy power inequality for 50:50 beamsplitter]
Consider the map $\cE=\cE_{1/2}$ associated with a 50:50 beamsplitter. Let $\rho_X,\rho_Y$ be two $n$-mode states. Then
\begin{align}
e^{S(\cE(\rho_X\otimes\rho_Y))/n}\geq \frac{1}{2}e^{S(\rho_X)/n}+\frac{1}{2}e^{S(\rho_Y)/n}\  .\mylabel{eq:entropypoweruseful}
\end{align}
\end{theorem}
\begin{proof}
 The proof follows Blachman~\cite{Blachman65}, and is reproduced here for the reader's convenience.
Define the functions 
\begin{align}
\begin{matrix}
F&\mapsto E_X(F)&:=&\exp\left({S(e^{F\cL}(\rho_X))/n}\right)\\
G&\mapsto E_Y(G)&:=&\exp\left({S(e^{G\cL}(\rho_Y))/n}\right)\\
H&\mapsto E_Z(H)&:=&\exp\left({S(e^{H\cL}(\cE(\rho_X\otimes\rho_Y)))/n}\right)
\end{matrix}\label{eq:entropypowerdiffusedstates}
\end{align}
expressing the entropy powers of the states obtained by letting diffusion act on~$\rho_X,\rho_Y$ and $\cE(\rho_X\otimes\rho_Y)$ for times $F,G$ and $H\geq 0$, respectively. According to Corollary~\ref{cor:scaling}, these functions have identical scaling 
\begin{align}
E(T)\sim const.+Te/2\qquad\textrm{ for }T\rightarrow\infty\ .\label{eq:epasymptotics}
\end{align} 

Moreover, since $E_X$ and $E_Y$ are continuous functions on the positive real axis, Peano's existence theorem shows
that the  initial value problems
\begin{align}
\begin{matrix}
\dot{F}(t)&=&E_X(F(t))\ ,\qquad &F(0)&=&0\\
\dot{G}(t)&=&E_Y(G(t))\ ,\qquad &G(0)&=&0
\end{matrix}\label{eq:initialvalueproblems}
\end{align}
have (not necessarily unique) solutions $F(\cdot),G(\cdot)$. Fix a pair of solutions and set 
\begin{align}
H(t)=(F(t)+G(t))/2\ .\label{eq:hhdef}
\end{align} Since 
\begin{align*}
\dot{F}(t)\geq e^{g(F(t)-1)}> 0\qquad\textrm{ and }\qquad \dot{G}(t)\geq e^{g(G(t)-1)}>0
\end{align*}
according to Theorem~\ref{thm:scaling}, these functions diverge, i.e., 
\begin{align}
\lim_{t\rightarrow\infty}F(t)=\lim_{t\rightarrow\infty}G(t)=\lim_{t\rightarrow\infty}H(t)=\infty\ .\label{eq:fghdivergence}
\end{align}
Consider the function
\begin{align*}
\delta(t)&=\frac{E_X(F(t))+E_Y(G(t))}{2E_Z(H(t))}\ .
\end{align*}
With the initial conditions~\eqref{eq:initialvalueproblems}, it follows that the claim~\eqref{eq:entropypoweruseful} is equivalent to
\begin{align}
\delta(0)\geq 1\ .\label{eq:toshowinequalitystam}
\end{align}
Inequality~\eqref{eq:toshowinequalitystam} follows from two claims: first, we have
\begin{align*}
\lim_{t\rightarrow\infty}\delta(t)=1\ ,
\end{align*}
as follows immediately from the asymptotic scaling~\eqref{eq:epasymptotics} of the entropy powers~$E$, as well as~\eqref{eq:fghdivergence} and the choice~\eqref{eq:hhdef} of~$H(t)$. Second, we claim that the derivative of $\delta$ satisfies
\begin{align}
\frac{d}{dt}\delta(t)\geq 0\qquad\textrm{ for all }t\geq 0\ .\label{eq:derivativecomputation}
\end{align}

Define, in analogy to~\eqref{eq:entropypowerdiffusedstates}, the Fisher informations of the states $\rho_X,\rho_Y$ and $\cE(\rho_X\otimes\rho_Y)$ after diffusion for times $F,G$ and $H\geq 0$ as
\begin{align}
\begin{matrix} 
J_X(F)&:=&J\left(e^{F\cL}(\rho_X)\right)\\
J_Y(G)&:=&J\left(e^{G\cL}(\rho_Y)\right)\\
J_Z(H)&:=&J\left(e^{H\cL}(\cE(\rho_X\otimes\rho_Y))\right)\ .
\end{matrix}\label{eq:fisherinfodiffusedstates}
\end{align}
Using the fact that $\{e^{T\cL}\}_T$ is a semigroup, we can restate the quantum Bruijin identity~(Theorem~\ref{thm:debruijin}) as
\begin{align}
\dot{E}_A(T)=\frac{1}{4n} E_A(T)J_A(T)\qquad\textrm{ where }\qquad A\in\{X,Y,Z\}\ \label{eq:fisherinforamtioncompact}
\end{align}
Note also that $J_X(F), J_Y(G)$ and $J_Z(H)$ are related by Stam's inequality if $H=(F+G)/2$ is the average~\eqref{eq:hhdef} of $F$ and $G$: because
\begin{align*}
e^{H\cL}(\cE(\rho_X\otimes\rho_Y))=\cE(e^{F\cL}(\rho_X)\otimes e^{G\cL}(\rho_Y))\ ,
\end{align*}
by the compatibility of the beamsplitter with diffusion (Lemma~\ref{lem:differentiation}), 
we obtain
\begin{align}
J_Z(H)\leq \frac{2J_X(F)J_Y(G)}{J_X(F)+J_Y(G)}\ \label{eq:staminequality}
\end{align}
using the quantum Stam inequality (Corollary~\ref{cor:staminequality})
and the nonnegativity of divergence-based Fisher information. 
Computing the derivative of $\delta$,
\begin{align*}
\begin{split}
\dot{\delta}(t)&=\frac{\dot{E}_X(F(t))\dot{F}(t)+\dot{E}_Y(G(t))\dot{G}(t)}{2E_Z(H(t))}\\
&\qquad\qquad\qquad -
\frac{E_X(F(t))+E_Y(G(t))}{2E_Z(H(t))^2}\cdot \dot{E}_Z(H(t))\dot{H}(t)\ , 
\end{split}
\end{align*}
inserting~\eqref{eq:fisherinforamtioncompact} and suppressing the $\delta$-dependence leads to
\begin{align*}
8n \dot{\delta}&=
E_X(F)J_X(F)\dot{F}+E_Y(G)J_Y(G)\dot{G}\\
&\qquad\qquad -(E_X(F)+E_Y(G))\cdot E_Z(H)J_Z(H)\dot{H}
\\
&=E_X(F)^2 J_X(F)+E_Y(G)^2J_Y(G)\\
&\qquad\qquad -\frac{1}{2}(E_X(F)+E_Y(G))^2J_Z(H)\ ,
\end{align*}
where we inserted the definition~\eqref{eq:initialvalueproblems} of $F$ and $G$. Replacing~$J_Z(H)$ by the upper bound~\eqref{eq:staminequality} and further suppressing dependencies finally gives
\begin{align*}
8n\dot{\delta}&\geq E_X^2J_X+E_Y^2 J_Y-(E_X+E_Y)^2 \frac{J_XJ_Y}{J_X+J_Y}\\
&=\frac{(E_XJ_X-E_YJ_Y)^2}{J_X+J_Y}\ ,
\end{align*}
hence~\eqref{eq:derivativecomputation} follows from the nonnegativity of the Fisher information. 
\end{proof}

\appendix
\section{The Q-function: Definition and Properties\mylabel{sec:qfunction}}
Define the $Q$-function in terms of the characteristic function by the expression
\begin{align}
Q(\xi)&=\frac{1}{(2\pi)^{2n}} \int e^{-iJ\xi\cdot \eta}\chi(\eta)e^{- \|\eta\|^2/4} d^{2n}\eta\ .\mylabel{eq:Qfunctiondefinition}
\end{align}
Identity~\eqref{eq:Qfunctiondefinition} can be inverted to give
\begin{align}
\chi(\eta)&=e^{\|\eta\|^2/4}\int e^{i\eta\cdot J\xi} Q(\xi)d^{2n}\xi\ .\mylabel{eq:characteristicfunctionF}
\end{align}
Inserting the definition~\eqref{eq:characteristicfunctiondef} of the characteristic function into~\eqref{eq:Qfunctiondefinition} gives
\begin{align*}
Q(\xi)&=\frac{1}{(2\pi)^n}\tr\left(\rho\frac{1}{(2\pi)^{n}} \int e^{i\eta\cdot J\xi}D(-\eta)e^{- \|\eta\|^2/4}d^{2n}\eta\right)
\end{align*}
Comparing this with~\eqref{eq:characteristicfunctiongaussian}, we conclude that
\begin{align*}
Q(\xi)&=\frac{1}{(2\pi)^n} \bra{\xi}\rho\ket{\xi}\ ,
\end{align*}
where~$\ket{\xi}=D(\xi)\ket{0}$ is a coherent state.

 This shows that
$Q(\xi)\geq 0$ for all $\xi\in\mathbb{R}^{2n}$. Furthermore, since
the coherent states satisfy the completeness relation
\begin{align*}
I&=\frac{1}{(2\pi)^n}\int\proj{\xi} d^{2n}\xi\ ,
\end{align*}
 we have $\int Q(\xi) d^{2n}\xi=1$, i.e., $Q$ is a probability density. 

The state~$\rho$ can be expressed in terms of $Q$  using~\eqref{eq:characteristicfunctionF} and~\eqref{eq:rhocharacteristicfunction}: we have
\begin{align}
\rho&=\int Q(\eta)\sigma(\eta)d^{2n}\eta\qquad\textrm{ where }\nonumber\\
\sigma(\eta)&=\frac{1}{(2\pi)^n}\int e^{\|\xi\|^2/4}e^{i\xi\cdot J\eta} D(-\xi)  d^{2n}\xi\ . \mylabel{eq:decompositionhusimi}
\end{align}

\section*{Acknowledgments}
The authors  thank Mark Wilde for his comments on the manuscript. G.M. acknowledges support by DARPA QUEST program under contract no.~HR0011-09-C-0047. R.K. acknowledges support by NSERC and IBM Watson research, where most of this work was done.


\end{document}